\documentclass[reprint,
notitlepage,
superscriptaddress,
 amsmath,amssymb,
 aps,
 prx,
]{revtex4-1}

\usepackage{nicefrac}
\usepackage{enumerate}
\usepackage{algorithm}
\usepackage{algpseudocode}
\usepackage{dirtytalk}
\usepackage{dsfont}
\usepackage{verbatim}
\usepackage{graphicx}
\usepackage{dcolumn}
\usepackage{bm}
\usepackage{hyperref}

\usepackage{microtype} 
\usepackage[T1]{fontenc} 

\usepackage{mathrsfs}
\usepackage{amsthm}
\usepackage{xcolor}
\definecolor{darkgreen}{RGB}{50, 150, 0}
\definecolor{darkred}{RGB}{102, 0, 0}

\usepackage{mathtools} 

\usepackage{etoolbox} 

\makeatletter
\colorlet{phfcolor}{red!50!orange}
\colorlet{phfrmcolor}{phfcolor!50!gray!25!white}
\colorlet{phfrmcolorlink}{blue!40!phfrmcolor}
\def\phf{\@ifnextchar[\phf@c\phf@}
\def\phf@{\@ifstar\phf@s\phf@t}
\newcommand\phf@c[1][]{{\color{phfcolor}\fontfamily{cmbr}\fontseries{sb}\selectfont \;[\,{#1}\,]\;}}
\newcommand\phf@t[1]{{\color{phfcolor}{#1}}}
\newcommand\phf@s[1]{{\colorlet{docnotelinkcolor}{phfrmcolorlink}\color{phfrmcolor}{\itshape #1}}}
\robustify\phf
\makeatother

\usepackage{bbm}
\def\Ident{\id}

\newcommand{\cH}{\mathcal{H}}
\newcommand{\cB}{\mathcal{B}}
\newcommand{\cX}{\mathcal{X}}
\newcommand{\sC}{\mathscr{C}}
\newcommand{\sP}{\mathscr{P}}
\newcommand{\rmD}{\mathrm{D}}

\newcommand{\bU}{\mathbb{U}}

\newcommand{\ket}[1]{| \hspace{1pt} {#1} \rangle}
\newcommand{\ketbrad}[2]{| \hspace{1pt} {#1} \rangle \langle {#2} \hspace{1pt} |}
\newcommand{\ketbra}[1]{\ketbrad{#1}{#1}}
\newcommand{\braket}[2]{\langle {#1} \hspace{1pt} | \hspace{1pt} {#2} \rangle}
\newcommand{\bra}[1]{\langle {#1} \hspace{1pt} |}

\DeclarePairedDelimiterX\matrixel[3]{\langle}{\rangle}{%
  {#1}\hspace*{0.2ex}\delimsize\vert\hspace*{0.2ex}{#2}%
  \hspace*{0.2ex}\delimsize\vert\hspace*{0.2ex}{#3}%
}

\newcommand{\tr}{\mathrm{tr}}

\newcommand{\diff}{\mathrm{d}}

\DeclarePairedDelimiterX\norm[1]{\lVert}{\rVert}{{#1}} 

\newcommand{\id}{\mathds{1}}

\newcommand{\ip}[2]{\left\langle #1 , #2\right\rangle} 

\newtheorem{theorem}{Theorem}
\newtheorem*{theorem*}{Theorem}
\newtheorem{proposition}{Proposition}

\newcommand{\ChoiH}{\sC(\cH_{AB})}
\DeclareMathOperator{\End}{End}

\newcommand{\ChannelH}{\sC(\cH_A\to\cH_B)}

\begin{document}

\preprint{APS/123-QED}

\title{Practical and reliable error bars for quantum process tomography}

\author{Le Phuc Thinh}
\affiliation{QuTech, Delft University of Technology, Lorentzweg 1, 2628 CJ Delft, The Netherlands}
\affiliation{Centre for Quantum Technologies, National University of Singapore, 3 Science Drive 2, 117543, Singapore}

\author{Philippe Faist}
\affiliation{Institute for Quantum Information and Matter, California Institute of Technology, Pasadena CA, 91125, U.S.A.}

\author{Jonas Helsen}
\affiliation{QuTech, Delft University of Technology, Lorentzweg 1, 2628 CJ Delft, The Netherlands}

\author{David Elkouss}
\affiliation{QuTech, Delft University of Technology, Lorentzweg 1, 2628 CJ Delft, The Netherlands}

\author{Stephanie Wehner}
\affiliation{QuTech, Delft University of Technology, Lorentzweg 1, 2628 CJ Delft, The Netherlands}

\date{\today}

\begin{abstract}
Current techniques in quantum process tomography typically return a single point estimate of an unknown process based on a finite albeit large amount of measurement data. Due to statistical fluctuations, however, other processes close to the point estimate can also produce the observed data with near certainty. Unless appropriate error bars can be constructed, the point estimate does not carry any sound operational interpretation. Here, we provide a solution to this problem by constructing a confidence region estimator for quantum processes. Our method enables reliable estimation of essentially any figure-of-merit for quantum processes on few qubits, including the diamond distance to a specific noise model, the entanglement fidelity, and the worst-case entanglement fidelity, by identifying error regions which contain the true state with high probability.  We also provide a software package---QPtomographer---implementing our estimator for the diamond norm and  the worst-case entanglement fidelity. We illustrate its usage and performance with several simulated examples. Our tools can be used to reliably certify the performance of e.g.\@ error correction codes, implementations of unitary gates or more generally any noise process affecting a quantum system.
\begin{description}
\item[Usage]
Secondary publications and information retrieval purposes.
\item[PACS numbers]
May be entered using the \verb+\pacs{#1}+ command.
\item[Structure]
You may use the \texttt{description} environment to structure your abstract;
use the optional argument of the \verb+\item+ command to give the category of each item. 
\end{description}
\end{abstract}

\pacs{Valid PACS appear here}
\maketitle


\section{\label{sec:intro}Introduction}

  Quantum technologies are improving at an ever faster pace, not only by a
  concentrated academic effort but increasingly via collaborations with
  industry.  Quantum technologies require very precise manipulation and control
  of quantum systems, fueling the development of theoretical tools for precise
  calibration and characterization of quantum devices~\cite{Helstrom1969}. Notably, quantum state tomography and quantum process tomography (also known as quantum process tomography) can infer the quantum state or the quantum process that describes a quantum device, providing a natural ``quantum debugger''~\cite{Paris_book}.

Quantum state tomography aims to reconstruct the unknown state of a system with reference to a set of calibrated measurement apparatuses. Many reconstruction techniques---formally known as \emph{estimators}---and their statistical properties have been developed and understood. These estimators can be roughly categorized into two groups based on the information they return about the unknown state. Point estimators take tomographic data from experiments and return a \emph{single} quantum state, i.e.\@ a density matrix, that best approximates the true  unknown underlying physical state. Examples in this category are linear inversion and maximum likelihood estimators~\cite{Hradil97,ML16,ML17}. By contrast, region estimators return a \emph{set} of quantum states in order to account for the uncertainty associated with the reconstruction. For state tomography many region estimators have been constructed, for instance, confidence regions~\cite{CR12,FR16,BK12_regions} and Bayesian regions~\cite{Bayesian_regions1,Bayesian_regions2}. Good region estimators have the advantage of providing robust statements associated with any chosen failure probability, that is, one can control the level of confidence with which the statement is made.
Moreover, unlike point estimators, region estimators have sound operational interpretation under the influence of statistical fluctuations from finite data. Consequently, region estimators are suitable for the certification of quantum hardware for practical applications.

Many tools for quantum process tomography are adapted from quantum state
  tomography, for instance via the Choi-Jamio\l{}kowski state-process
  correspondence~\cite{processtomo}.  Beyond traditional process
  tomography~\cite{Chuang1997JMO}, there are also more advanced tools such as
  randomized
  benchmarking~\cite{Knill2008PRA_randbench,Chow2009,Kimmel2014,Jonas},
  gate-set tomography~\cite{BlumeKohout2017NComm_threshold} and compressed
  sensing~\cite{Kliesch2017arXiv_processes}, that display
  certain advantages, such as a reduced number of required measurements.  In the
  case of region estimators, some subtleties prevent a straightforward
  application of the corresponding tools for quantum states to quantum
  processes.  Indeed, the set of quantum process is in one-to-one
  correspondence with only a subset of all bipartite states, namely those whose
  reduced state on one system is maximally mixed; this constraint has to be
  incorporated explicitly in the region estimator. In this paper, we
  enrich the statistical toolbox for quantum process tomography by providing a
  confidence region estimator for quantum process inspired by the state tomography method of
  Christandl and Renner~\cite{CR12}.

Often in certifying specific applications, we are not interested in the full knowledge of the quantum process; a property of the unknown channel suffices. For example, in quantum key distribution we are often interested in how close the final state output by the protocol is to the ideal key-state; this is captured for instance by the fidelity or the trace distance of the real state to the ideal state~\cite{QKDSecurity}. Likewise, in quantum computing a relevant figure-of-merit that enables fault-tolerant computation is the error threshold captured by the diamond distance or the worst-case entanglement fidelity of the real implemented gate relative to the ideal gate~\cite{FaultTolerant}. Note, though, that even a single figure of merit may serve as a full characterization of a process: A bound on the diamond distance or the entanglement fidelity to a given fixed channel confines the true channel to a small region in channel space.  For these reasons, and because this significantly simplifies our analysis, we focus on estimators for quantum processes that report confidence intervals for a given figure of merit.

\textbf{Summary of main results:}

Our main contribution is three-fold:
\begin{enumerate}[(i)]
\item A confidence region estimator for channel tomography through the use of the Christandl-Renner-Faist estimator for states and the Choi-Jamiolkowski isomorphism between quantum states and quantum processs. We call this \emph{the biparite-state sampling method}.
\item A new confidence region estimator to \emph{directly} (without first tomographing the Choi state associated with the channel) estimate quantum processs and its proof of correctness. We call this \emph{the channel space sampling method}.
\item A software package called \emph{QPtomographer}~\cite{Our-Code-Github} accompanying our theoretical results for analysing experimental data. Our software returns \emph{quantum error bars} which captures all the information about the unknown channel derivable from the tomographic data and enables the user to construct confidence regions for any confidence level of interest.
\end{enumerate}

By comparing the differences of the two estimators, we obtain a better understanding about the relationship between probability measures on state space and channel space which may be of independent interest. Because the estimators return a confidence region, they will work without any assumption on the prior distribution of the unknown process.

To illustrate how to use our result, we consider the scenario of \emph{certifying a quantum memory} (an example of quantum property testing~\cite{propertytesting}). This corresponds to certifying that a quantum device (approximately) implements the identity channel. We consider three possible figures-of-merit: the diamond distance to the identity channel, the entanglement fidelity and the worst-case entanglement fidelity~\cite{Schumacher1996PRA_sending,Gilchrist2005PRA_processes}. Our method yields a reliable estimation of these figures-of-merit.

The paper is organized as follows. We first demonstrate in~\autoref{sec:workflow} how one can use our method to obtain reliable information in a tomography experiment. The correctness of our tools is justified in~\autoref{sec:main_results} where we present the main results. Then we study the behavior of our numerical implementations in~\autoref{sec:apply} before concluding our paper with future directions (\autoref{sec:conclusion}). We leave the formal statements and detailed derivations of our results to the Appendices.

\section{\label{sec:workflow}Setup and workflow}
In this section, we detail the main workflow associated with the tomographic tools we have developed in our paper via a concrete example. 

Suppose an experimental team has developed a working quantum memory (single qubit) and would like to certify its performance for usage within a quantum communication protocol such as entanglement distillation. In this context, one way of measuring the performance is the diamond norm distance to the identity process. The workflow for this example is illustrated in~\autoref{fig:workflow}. We remark that there are other quantities of interest which do not assume an i.i.d structure, such as for example estimating the capacity as in~\cite{Pfister2018}.

The quantum memory's performance can be determined as follows.  We assume that
  we have access to a given number of uses of the quantum memory.  The number of uses can
  be chosen freely, noting  that it affects the final error bars.

\begin{figure}
  \centering
  \includegraphics[width=\columnwidth]{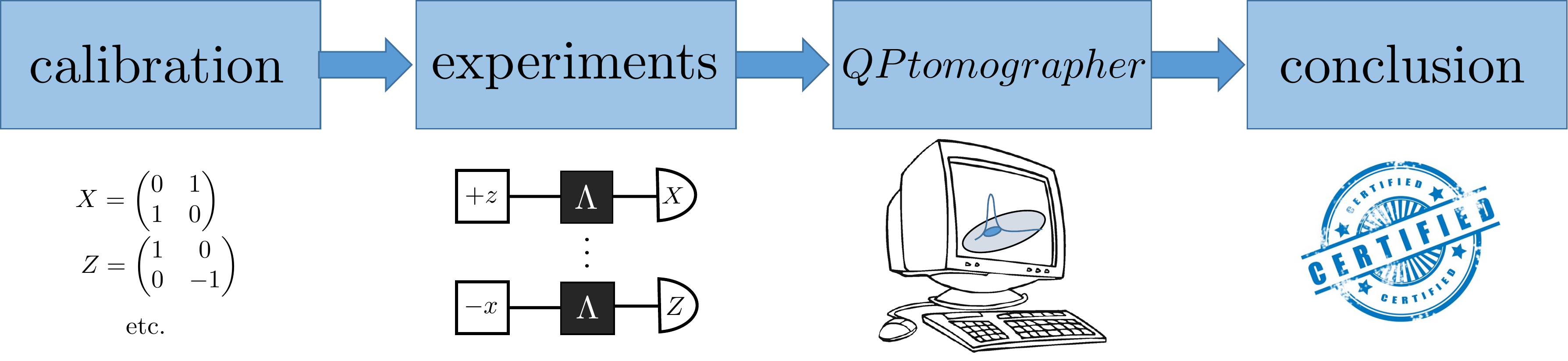}
  \caption{\label{fig:workflow}The workflow of rigorous process tomography. Our data analysis \emph{QPtomographer} supports both prepare-and-measure and ancilla-assisted experimental schemes. The conclusion is guaranteed without any prior information on the unknown quantum process.}
\end{figure}

\begin{figure} 
  \centering
  \includegraphics[width=0.5\columnwidth]{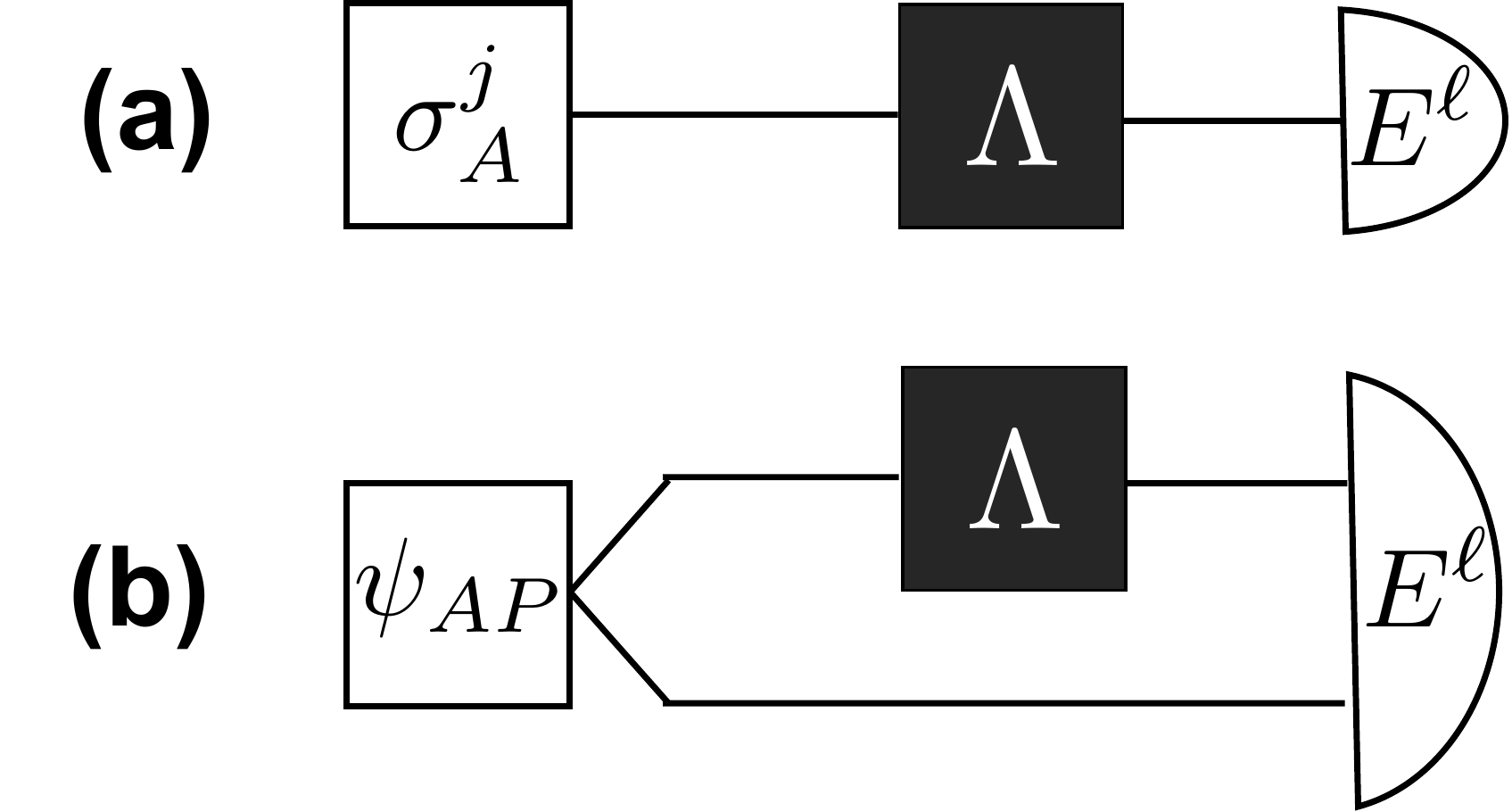}
  \caption{\label{fig:tomo}Illustrations of (a) prepare-and-measure and (b) ancilla-assisted tomographic schemes for an unknown channel $\Lambda_{A\to B}$. In prepare-and-measure, one can only prepare input state $\sigma_A^j$ which is fed into an unknown channel whose output state is measured by some POVM with elements $E_B^\ell$. In ancilla-assisted, one can prepared entangled input state with some reference system $P$, and measure jointly the output using some POVM with elements $E^\ell_{BP}. $}
\end{figure}

Moreover, in order to find out what the unknown process was, we need additional access to state preparation and measurement devices which are information-complete (at least in the physical degrees of freedom where the unknown process acts). In this example, the set of state preparations are the Pauli eigenstates $\ket{\pm x}, \ket{\pm y}, \ket{\pm z}$, while the set of measurement devices are Pauli $X, Y, Z$ measurements. We assume that each use of the quantum process are independent, and that the same unknown quantum process is applied for each run of the experiment, yielding statistics which are independent and identically distributed (i.i.d.). While here we consider a \emph{prepare-and-measure} scenario as depicted in~\autoref{fig:tomo}\textbf{(a)}, it is also possible to consider an \emph{ancilla-assisted} scheme (\autoref{fig:tomo}\textbf{(b)}).

The first step (see~\autoref{fig:workflow}) involves calibrating the state preparation and measurement devices to have $\ket{\pm x}, \ket{\pm y},\ket{\pm z}$ state preparations and $X, Y, Z$ measurements. After this calibration procedure has succeeded, one performs a chosen number $n=45000$ of individual experiments. Each experiment consists of the following steps
\begin{itemize}
\item Prepare an input state by executing one of the devices $\ket{\pm x}, \ket{\pm y}, \ket{\pm z}$ (perhaps at random).
\item Apply the (unknown) quantum memory to the said input state.
\item Measure the output state using one of the possible $X, Y, Z$ measurement devices (perhaps at random).
\item Record the outcome of this experiment in a dataset~$E$.
\end{itemize}
We remark that the preparation and measurement should yield sufficient data in the sense that all combination of input states and measurements should be chosen (perhaps at random).

Such a dataset $E$ can then be analyzed by our software \emph{QPtomographer}. One provides to our software the information about the measurement settings and the observed dataset. Then, using a Metropolis-Hastings sampling method, the software determines a specific type of distribution of the figure-of-merit~\autoref{fig:output} along with corresponding \emph{quantum error bars} $(v_0,\Delta,\gamma)$.  The value $v_0$ is the location of the maximum in~\autoref{fig:output}, while $\Delta$ and $\gamma$ measure the spread of the error. In our example, the analysis based on the input data set $E$ with $n=45000$ measurement records returned the quantum error bars
$$(v_0=0.058,\Delta=0.006,\gamma=0.00019),$$
which determine the parameters of an appropriate fit function (red curve of~\autoref{fig:output}).
The quantum error bars contain all the information about the error analysis.  Namely, they (i) form a concise description of the error, (ii) provide an intuitive idea of the magnitude of the error, and (iii) can easily determine confidence regions for the quantum state or quantum process~\cite{FR16}.  In this sense, quantum error bars are perfectly analogous to classical error bars: The latter are indeed a concise, intuitive description of the error from which one easily determines rigorous confidence intervals.  For this reason it is a natural object to report at the end of a process tomography procedure.

If one wishes to actually derive rigorous confidence regions for the diamond norm distance, one may proceed as follows. First, one fixes a  confidence level, say $\alpha=99\%$, which sets the corresponding error parameter as $\epsilon=1-\alpha=10^{-2}$. By Theorem~\ref{thm:channel-space}, for $n=45000$ (size of our dataset $E$) and $d_A=d_B=2$, we need to find a region of diamond norm distance values with weight at least
\begin{align*}
1- \frac{\epsilon}{2}\binom{2n+d^2_Ad^2_B-1}{d^2_Ad^2_B-1}^{-2} \geq 1-10^{-151}\,,
\end{align*}
With reference to~\autoref{fig:output}, this means we need to find the $x$-position such that the area under the curve exceeds $1-10^{-151}$. A numerical integration leads to a region at least as large as $[0,0.24]$. Together with the enlargement by
\begin{align*}
\delta = \sqrt{\frac{2}{n}\left(\ln\frac{2}{\epsilon}+3\ln \binom{2n+d^2_Ad^2_B-1}{d^2_Ad^2_B-1}\right)} = 0.1\,
\end{align*}
(to exclude nearby channels which could result in the same observed dataset with high probability) the final confidence region is $[0,0.34]$. This means we have certified that the diamond norm distance of the unknown quantum memory to an ideal quantum memory is at most $0.34$ with $99\%$ confidence. In general, increasing the number $n$ of measurement data points will shrink this confidence interval (due to the exponential decays in the diamond distance density, see also Appendix~\ref{sec:convergence}). 

\begin{figure}
  \centering
  \includegraphics[width=\columnwidth]{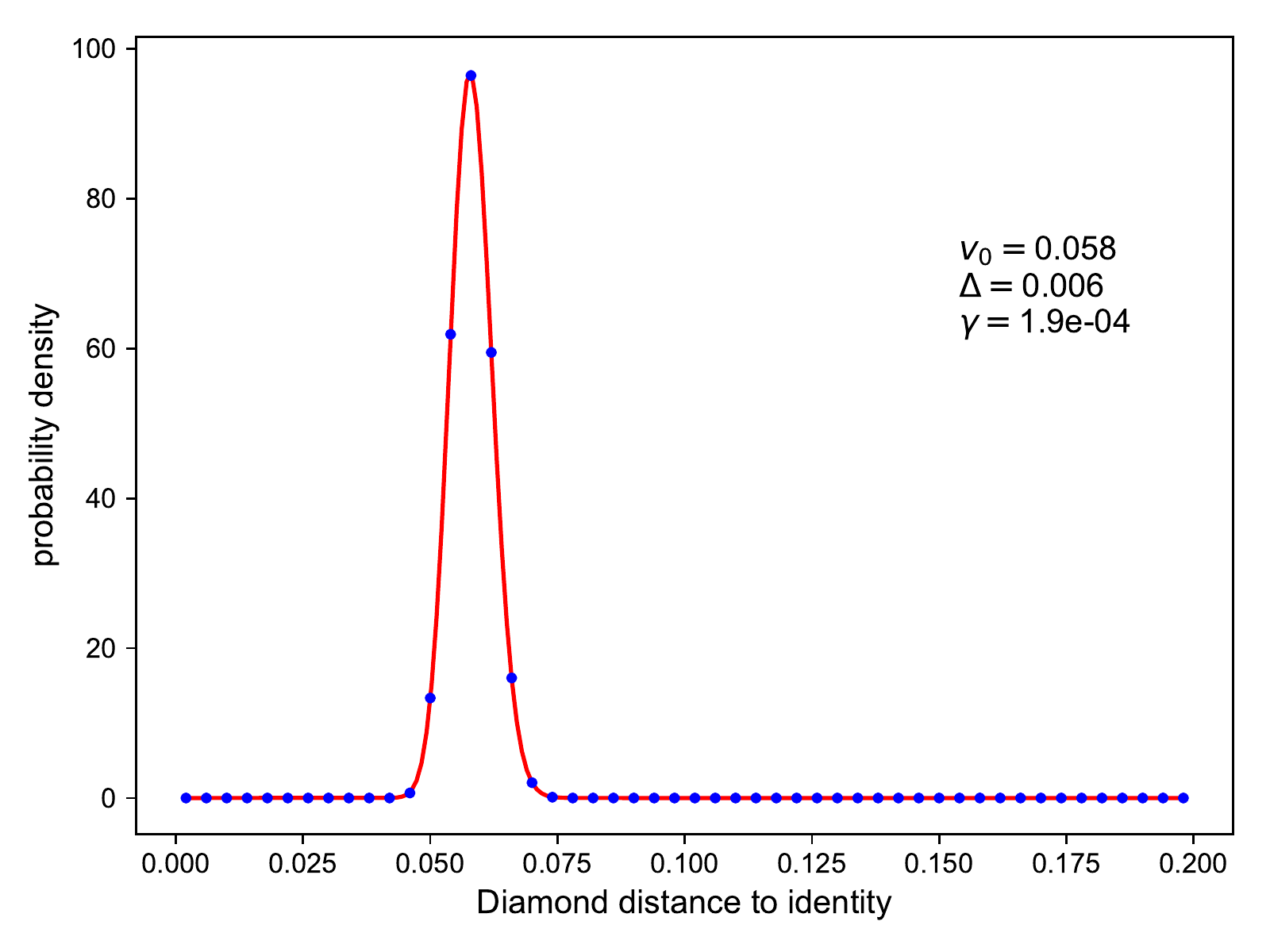}
  \caption{\label{fig:output}Typical output from \emph{QPtomographer}. The (blue) dots form the estimated distribution of values of the diamond norm distance to the identity channel as determined by the Metropolis-Hastings random walk. These are well-fitted to the (red) curve, which is compactly described by the triple of numbers $(v_0,\Delta,\gamma)$ which we called \emph{quantum error bars}. Here $v_0$ is the position of the peak, $\Delta$ is half width at relative height $1/e$, and $\gamma$ is a measure of skewness. These data encode information about the performance of the quantum memory, and enable us to construct confidence intervals certifying its quality.}
\end{figure}

  We emphasize that the unnaturally large size of the regions is due in large
  part to a technical difficulty in the proofs of our bounds that is dealt with
  by employing tools that are known not to be tight in this context.  For this
  reason, the quantum error bars are more informative than the actual final
  confidence regions.

This concludes the general workflow associated with our tomographic tools. The next section explain at a high level how our software transform tomographic data into confidence regions.

\section{\label{sec:main_results}Main results}

Our software package \emph{QPtomographer} is built on top of two rigorously proven theoretical constructions. These are confidence region estimators based on the bipartite-state sampling method or the channel-space sampling method. The bipartite-state method works in the ancilla-assisted scheme, while the channel-space method works in both ancilla-assisted and prepare-and-measure schemes. This section gives a high level overview of the constructions together with the main ideas behind the proof of correctness, and leave the details to Appendix~\ref{sec:naive} and Appendix~\ref{sec:main}, respectively. We begin with a brief motivation for confidence region estimators.

\subsection{Confidence region estimators}

In the limit of infinite data (i.e.\@ the number of records in dataset $E$ is infinity), it is possible to exactly compute the probabilities of each measurement outcome from $E$ and reconstruct the unknown channel by linear inversion on these observed probabilities~\cite{Paris_book}. However, in the practical scenario of finite data (i.e.\@ dataset $E$ contains $n$ records) statistical fluctuations will imply the failure of all point estimation methods such as linear inversion or maximum likelihood estimation. This is due to the fact that channels close to the point estimate can produce the same dataset with high probability.

In order to make statistically rigorous and operationally sound statements on the unknown channel in this regime, we turn to region estimators, which are generalisations of the process of constructing error bars. We will look at a type of region estimators known as \emph{confidence} region estimators. These are maps from data $E$ to subsets $S_E\subseteq\ChannelH$ of the set of quantum processes with the property that for all $\Lambda\in\ChannelH$
\begin{align}
\Pr_E[\Lambda\in S_E]\geq\alpha,
\end{align}
where $\alpha$ is a prefixed confidence level and the probability is evaluated over the random data $E$ according to the distribution $\Pr[E|\Lambda]$. It is important to note that confidence is a property of the entire estimator (the procedure $E\mapsto S_E$) and not of any particular subset $S_E$ produced by the estimator.

The operational meaning of confidence region estimators can be understood as follows. Suppose the black box implementing the unknown channel $\Lambda_{\mathrm{true}}$ is in fact prepared by a referee, who knows exactly which channel the black box applies. We proceed with a sequence of state preparations, applications of the channel and measurements of the output states to obtain a dataset $E$. Then we apply the estimator on $E$ to get $S_E$. Repeating this procedure a large number of times, say $N=10^5$, if $r$ denotes different repetitions then we obtain different datasets $E(r=1),...,E(r=10^5)$ with corresponding conclusions that the true channel $\Lambda_\textrm{true}$ should be in the region $S_{E(r=1)},...,S_{E(r=10^5)}$. Now since the referee knows exactly the unknown channel, the referee can evaluate the \emph{proportion of correct conclusions}
\begin{align*}
\frac{\left|\{r:r=1, ..., N \text{ and } \Lambda\in S_{E(r)} \text{ is true }\}\right|}{N}\,.
\end{align*}
If the estimator used is a confidence region estimator with confidence level $\alpha=0.99$, then in the limit of $N\rightarrow\infty$ this proportion is at least $0.99$. This is the meaning of confidence: the correct conclusion is guaranteed for a large number of uses of the estimator, regardless of the unknown channel. Note that for a specific use of the estimator which returns $S_E$, we \emph{cannot} draw the conclusion that $\Lambda\in S_E$.

An alternative justification of confidence regions comes from a Bayesian point of view: Bayesian tomography uses outcomes of measurements to update a prior distribution about the quantum state to a posterior distribution.  While this posterior clearly depends on the prior, it is known that when enough data is collected, the posterior distribution is no longer sensitive to the exact prior which was originally used (as long as the original prior has full support).  Now consider a high-weight region of a posterior distribution, which is also known as a \emph{credible region}.  We may ask to what extent this region remains a credible region if we change the underlying prior. It turns out that for a large enough number of measurements, 
   we may find regions which are credible regions for \emph{any} prior, except for some exceptionally unlikely measurement datasets~\cite{CR12}.  Such regions are precisely confidence regions.  

\subsection{\label{sec:main_estimators}Our confidence region estimators}

Our method of constructing region estimators uses the information about the underlying unknown channel via the \emph{likelihood function} defined generically for an observed dataset $E$ as
\begin{align}
\mathcal{L}(\Lambda|E) = \Pr(E|\Lambda),
\end{align}
where the probability of the dataset $E$ under the assumption that the unknown channel is $\Lambda$ is given by Born's rule. The specific form of the likelihood function depends on the scenarios and assumptions we postulate, c.f.\@ Appendices~\ref{sec:naive},\ref{sec:main}. The likelihood function can be seen as giving a ranking about which channel best produce the observed dataset. We now present our methods of process tomography.

{\bf Bipartite-state sampling method:} the main idea behind this method is that a quantum process is in correspondence with bipartite Choi states via the Choi-Jamiolkowski isomorphism. Hence, we can construct confidence regions for quantum states using the method of Christandl-Renner, and then perform an additional classical post-processing step to recover a confidence region for quantum processes.

Let us now first assume the use of an ancilla-assisted tomographic scheme~\autoref{fig:tomo}\textbf{(b)}, which loosely corresponds to physically performing the Choi-Jamiolkowski isomorphism in the laboratory. This means having access to a full rank bipartite entangled state $\ket{\psi_{AP}}$ as input to the channel, and performing tomography on the output state $\rho_{BP}:=\Lambda_{A\to B}(\psi_{AP})$ which is the unknown Choi state associated with the unknown channel. 

Treating $\rho_{BP}$ as the unknown state in a state tomography problem, we now apply the Christandl-Renner method of constructing confidence regions from tomographic data. Recall that the Christandl-Renner confidence region is constructed from the measure
\begin{align}
\diff \mu_E(\sigma_{AB}) := c_E^{-1}\tr(\sigma_{AB}^{\otimes n} E)\diff\sigma_{AB}
  \label{eq:def-mu-E-sigmaAB}
\end{align}
where $c_E=\int\tr(\sigma_{AB}^{\otimes n} E)\diff\sigma_{AB}$ is the normalizing constant, and $\diff\sigma_{AB}$ is the uniform distribution on bipartite density matrices (obtained by tracing out a Haar random pure state on a larger space). Note that $\tr(\sigma_{AB}^{\otimes n} E)$ is the likelihood function for the outcome $E$ given the state $\sigma_{AB}$ in this scenario. Confidence regions for the unknown $\rho_{AB}$ can be constructed from $\diff\mu_E(\sigma_{AB})$ as the following proposition asserts.
\begin{theorem}[Christandl \& Renner~\cite{CR12}, informal]
\label{thm:CRtomo}
Let $n$ be the number of systems measured by a POVM during tomography and $1-\epsilon$ be the desired confidence level. Let $S_{\mu_E}\subseteq\rmD(\cH_{AB})$ be any set of bipartite states with high weight under the probability measure $\diff\mu_E(\sigma_{AB})$. Then the enlargement in purified distance $S_{\mu_E}^\delta$ where
\begin{align}\label{eq:maintext_enlargement_bipartite}
\delta = \sqrt{\frac{2}{n}\left(\ln\frac{2}{\epsilon}+2\ln s_{2n,d^2_{AB}}\right)}\,
\end{align}
with $s_{n,d}:=\binom{n+d-1}{d-1}$ is a confidence region of confidence level $1-\epsilon$.
\end{theorem}

Intuitively, we can think of the enlargement as a way to exclude nearby states/channels (relative to a proposed region of states/channels) that can give rise to the same observed dataset $E$ with nonzero probability.

The confidence region $S_{\mu_E}^\delta$ contains bipartite quantum states which are not Choi states. This is due to the fact that the method of Christandl and Renner does not \emph{a priori} allow the Choi state constraint $\tr_{B}(\sigma_{AB})=\id_A/d_A$. Hence, we have to invent an additional post-processing step to map $S_{\mu_E}^\delta$ to a region consisting of exclusively Choi states. By the Choi-Jamiolkowski isomorphism we then have a confidence region for the unknown quantum process. The detailed explanation is left to Appendix~\ref{sec:naive}.\\

{\bf Channel-space sampling method:} this method is a new construction of confidence region that directly returns channel-space confidence regions. Compared to the bipartite-state method, the channel-space method works in both the prepare-and-measure and ancilla-assisted tomographic schemes and takes into account the \emph{a priori} knowledge that we are estimating a quantum process. This leads to computational efficiency relative to the bipartite-state method because the additional post-processing step of the bipartite-state method is not required here.

The estimator is constructed from the probability measure on the set of quantum processs $\ChannelH$
\begin{align}
\diff\nu_E(\Lambda):=c'^{-1}_E\mathcal{L}(\Lambda|E)\diff \nu(\Lambda)
\label{eq:def-nu-E-channels}
\end{align}
where $\mathcal{L}(\Lambda|E)$ is the likelihood for the event $E$ given a channel $\Lambda$, $c'_E=\int\mathcal{L}(\Lambda|E)\diff \nu(\Lambda)$ serves as a normalizing constant and $\diff \nu(\Lambda)$ is the Haar-induced measure on $\ChannelH$. The likelihood function is adapted depending on prepare-and-measure or ancilla-assisted tomographic scheme and is defined as the probability of obtaining the dataset $E$ given a channel $\Lambda$. Informally, this measure captures the information of the unknown channel as revealed by the observed dataset $E$ in an unbiased manner (that is without using any prior knowledge on the unknown).

Given this measure, we obtain
 \begin{theorem}[informal]
 \label{thm:channel-space}
Let $n$ be the number of channel uses during tomography and $1-\epsilon$ be the desired confidence level. Let $R_{\nu_E}\subseteq\ChannelH$ be a set of channels with high weight under the probability measure $\diff \nu_E(\Lambda)$. Then the enlargement in purified distance (for quantum processs, induced from states) $R_{\nu_E}^\delta$ where
\begin{align}\label{eq:maintext_enlargement_channel}
\delta = \sqrt{\frac{2}{n}\left(\ln\frac{2}{\epsilon}+3\ln s_{2n,d^2_{AB}}\right)}\,.
\end{align}
with $s_{n,d}:=\binom{n+d-1}{d-1}$ is a confidence region with confidence level $1-\epsilon$.
\end{theorem}

\emph{Confidence interval for figures-of-merit:} in practice, we choose the region in \autoref{thm:channel-space} for any chosen figure-of-merit to be the subset of channels whose figure-of-merit is better than a certain threshold. For the diamond norm distance to the ideal channel, we consider
\begin{align}
R = \left\{\Lambda:\nicefrac{1}{2}\norm{\Lambda-\Lambda^{\mathrm{ideal}}}_\diamond \leq \gamma_E\right\}\,,
\end{align}
and for the worst-case entanglement fidelity we consider
\begin{align}
R = \{\Lambda:F_{\mathrm{worst}}(\Lambda) \geq \gamma_E\}\,.
\end{align}
We can work directly with the figure-of-merit by push-forwarding the measure $\diff\nu_E(\Lambda)$ to the space of figures-of-merit, which is typically the reals $\mathbb{R}$ or the interval $[0,1]$, and obtain the histogram $h(v)$ over different values of the figure-of-merit; the enlargement of these regions under the purified distance is translated into a loss in the value of the figures-of-merit: $\gamma_E\to\gamma_E+d_A\delta/2$ for diamond distance and $\gamma_E-d_A\delta$ for worst-case entanglement fidelity. The loss vanishes with increasing number of channel uses (as evident in~\autoref{eq:maintext_enlargement_bipartite} and~\autoref{eq:maintext_enlargement_channel}), which allows reliable estimation of the figure-of-merit.


\subsection{Numerical implementations}

The previous section outlined the theoretical results underpinning our software package. We observe a reduction from the problem of constructing confidence regions to a problem of approximating the measures $\diff\nu_E(\Lambda)$ or $\diff \mu_E(\sigma_{AB})$. Solving this latter problem is the objective of the numerical implementations.

{\bf Computing $\diff\nu_E(\Lambda)$ and $\diff \mu_E(\sigma_{AB})$:} in order to approximate a probability measure, we will take the Monte-Carlo approach of producing its samples, i.e.\@ producing a histogram approximating a measure. More samples lead to better approximation but require more computational resources. Sampling according to $\diff \mu_E(\sigma_{AB})$ (i.e. the biparite-state method) has been implemented in~\cite{FR16}, and sampling according to $\diff\nu_E(\Lambda)$ (i.e. the channel space method) can be obtained by similar methods. More precisely, $\diff\nu_E(\Lambda)$ can be approximated by Metropolis-Hastings sampling~\cite{MHalgorithm} on channel space, which reduces to the ability of sample a ``uniformly random quantum process'' according to $\diff \nu(\Lambda)$. To do this, it suffices to sample a unitary operator at random according to the Haar measure, by Stinespring dilation (see Appendix~\ref{sec:measures}). Crucially, because we use the Metropolis-Hasting algorithm, it is not necessary to calculate the normalizing constants $c_E$ and $c'_E$ which are difficult to obtain in practice. The parameters required to run the Metropolis-Hastings algorithm are the initial starting point and a jump distribution (a distribution from which we know how to produce samples). For the jump distribution, we have implemented two versions which we call $e^{iH}$ and elementary rotation.  

The Metropolis-Hastings algorithm starts with an initial point $U_0$ in the sample space, which we take to be the identity unitary operator, and conducts a random walk around this space. For each iteration, starting from current location $U$ the jump distribution produces a candidate $U'$ (depending on the current location) for a sample---a unitary matrix---which could potentially comes from $\diff \nu_E(\Lambda)$. This candidate is accepted to be a sample of $\diff \nu_E(\Lambda)$ with acceptance probability $a$, and upon acceptance the current location is updated to this point. The acceptance probability is defined to be the likelihood ratio (i.e.\@ probability ratio) of $U'$ to produce the observed dataset $E$ with respect to the the current location $U$. This can be computed as the state preparations and measurements are known from calibration, and the dataset $E$ is given from the experiment. The sequence of points $\{U_i\}$ visited in this fashion, albeit correlated, are asymptotically distributed according to $\diff \nu_E(\Lambda)$~\cite{MHalgorithm}.

{\bf Extracting information for a given figure-of-merit:} in terms of a given figure-of-merit $f$, the distribution $\diff \nu(\Lambda)$ can be represented as a density function $h(v)$ for any possible value $v$ of the figure-of-merit associated with the unknown channel $\Lambda$. For all practical purposes, our goal is to obtain a compact description of this density. Clearly, this function is well approximated by the sequence of values $\{f(U_i)\}$ derived from the output of the Metropolis-Hastings algorithm by simply evaluating the figure-of-merit at each point $U_i$. We organize $\{f(U_i)\}$ into bins of some size to produce a histogram approximating $h(v)$. This histogram is further subjected to a statistical fit analysis to obtain quantum error bars $(v_0,\Delta,\gamma)$, which contain enough information to reconstruct a good approximation of $h(v)$.

We consider two fit models in this paper. The fit model given in Ref.~\cite{FR16}
\begin{align}
  \ln \mu^{\mathrm{fit,\#1}}(v) = -a_2 v^2 - a_1 v + m\,\ln v + c\
  \label{eq:fit-model-1}
\end{align}
does not have great agreement in our numerical examples (\autoref{sec:apply}) to the histogram bins. This leads us to develop an empirical model
\begin{align}
  \ln \mu^{\mathrm{fit,\#2}}(v) = -a_2 v^2 - a_1 v + m\,(\ln v)^p + c\ ,
\end{align}
which fits better to our examples (\autoref{sec:apply}). In any case, it is important to note that the functions $\mu(v)$ and $h(v)$ both decay exponentially fast (for the same reasons as in Ref.~\cite{FR16}).  Hence, when trying to find high-weight regions it is not crucial to know the shape of the function exactly; rather, any imprecision on the shape of the function incurring an error on the estimated weight of a region, can be compensated by only a small increase in the region size (a property of the exponential function). Hence, whenever unspecified, we report quantum error bars as given using fit model \#1 and as presented in Ref.~\cite{FR16}, keeping in mind that in a paranoid setting one would have to adjust the confidence regions for the corresponding error.  In summary, the reported quantum error bars are computed from the fit parameters of the fit model \#1 as:
\begin{subequations}
  \begin{align}
    v_0 &= \frac1{4a_2}\left[-a_1 + \sqrt{a_1^2+8 a_2 m}\right]\ ; \\
    \Delta &= \left(a_2 + \frac{m}{2 v_0^2}\right)^{-1/2}\ ; \\
    \gamma &= m\,\frac{\Delta^4}{6 v_0^3}\ .
  \end{align}\label{eq:qerrorbars}
\end{subequations}

See Appendix~\ref{sec:numerics} and~\autoref{sec:apply} for more details.

\subsection{Relation between our two sampling methods}

There is a connection between our two estimators, which we explain in detail in Appendix~\ref{sec:compare}. The essential difference between the bipartite sampling method and the channel-space method can be traced back to how one uses the prior information about the input state. In the former, nothing is assumed about the exact input state other than what can be inferred directly from the measurement data (of course, still under the physical  assumption of a pure entangled input); in the latter, the exact input state is assumed with certainty, and is used in the construction of the estimator (as manifestly visible in the likelihood function).

\section{\label{sec:apply}Application: Examples}
\subsection{One-qubit example}

We now illustrate in more details the use of our software package \emph{QPtomographer} by continuing the quantum memory example. The generic procedure is described in Algorithm~\autoref{ancilla-assisted} and Algorithm~\autoref{prepare-measure}.
\begin{algorithm}[H]
    \caption{quantum process Tomography}
    \label{tomo_experiment}
    \begin{algorithmic}[1] 
        \State Perform data collection via Algorithm 2 or Algorithm 3
        	\State Generate random samples (channel space or bipartite)
        	\State Compute histogram of figure-of-merit
        	\State Fit analysis of histogram
        \State \textbf{return} quantum error bars
    \end{algorithmic}
\end{algorithm}

\begin{algorithm}[H]
    \caption{Ancilla-Assisted (see~\autoref{fig:tomo}\textbf{(b)})}
    \label{ancilla-assisted}
    \begin{algorithmic}[1] 
        \State \textbf{input} a pure entangled state and a collection of measurements
        \For{$i=1$ to $n$}
        	\State Choose a measurement from the set
        	\State Apply the channel to the entangled input state
        	\State Measure the output state with the chosen measurement
        	\State Record the observed outcome
        \EndFor
        \State \textbf{return} dataset $E$ storing the measurement and outcomes for each repetition
    \end{algorithmic}
\end{algorithm}

\begin{algorithm}[H]
    \caption{Prepare-and-Measure (see~\autoref{fig:tomo}\textbf{(a)})}
    \label{prepare-measure}
    \begin{algorithmic}[1] 
        \State \textbf{input} a set of states and  a collection of measurements
        \For{$i=1$ to $n$}
        	\State Choose an input state and a measurement from the set
        	\State Apply the channel to this input state
        	\State Measure the output state with the chosen measurement
        	\State Record the input choice and the observed outcome
        \EndFor
        \State \textbf{return} dataset $E$ storing input state, output measurement and outcomes for each repetition
    \end{algorithmic}
\end{algorithm}

The output of our classical data analysis is called ``quantum error bars'' which contain all the information about the figure-of-merit that can be obtained from the tomographic dataset. From here, it is easy to construct confidence regions for any specified confidence level.\\

\emph{Step 1. Data collection:}

Consider the scenario of testing the performance of a quantum memory $\Lambda_{A\to B}$. The ideal channel we wish to implement is the identity channel $\mathcal{I}$. Suppose that the real channel implemented in the experiment the depolarizing channel
\begin{align}
  \Lambda_{A\to B}(\rho) = p\,\rho + (1-p)\,d_B^{-1}\id_B\ ,
\end{align}
acting on one qubit ($d_A=d_B=2$), with the parameter $p=0.9$. In other words, the experiment is slightly off from the ideal implementation by some  white noise.

Furthermore, we consider the ancilla-assisted scheme, and assume that the input to the channel is half of a pure entangled state $\ket\psi_{AP} = (\sigma_A^{1/2}\otimes\Ident)\, d_A^{1/2}\ket{\hat\Phi}_{AP}$, where we choose
\begin{align}
  \sigma_A = 
  \begin{pmatrix}
    0.6 & 0.1 \\
    0.1 & 0.4 \\
  \end{pmatrix}
  \ ,
\end{align}
which mimics an input state which deviates slightly from the maximally mixed state.  Note that the entangled input state has full Schmidt rank.

Since we do not have an actual experiment, we have to simulate Pauli measurements on the joint state $\rho_{BP}$ after application of the channel $\Lambda_{A\to B}$, with $2$ possible outcomes for each of the $3$ measurement settings.  For each measurement setting, 500 measurement outcomes were simulated. These constitutes the information contained in the (simulated) observed dataset $E$ with $n=45000$.

We now subject this dataset to an analysis which we aim to measure three figures-of-merit corresponding to our unknown channel: the diamond distance  to the identity channel, the average entanglement fidelity and the worst case entanglement fidelity. Refer to the Appendix~\ref{sec:prelim} for the precise definitions.\\

\emph{Step 2 and 3. Random sampling and histogram:}

We use the methods developed in~\autoref{sec:main_estimators} to estimate the three figures-of-merit.  The calculation of all three functions was done in C++ using the SCS toolbox~\cite{SCS1,SCS2}.  A simple Python interface was used to control the execution of the program.  All numerics were run on a 2016 Macbook Pro with 4 physical/8 virtual cores using our code provided at~\cite{Our-Code-Github}.

First, we demonstrate the \emph{bipartite-state method} described in Appendix~\ref{sec:naive}.  This consists in running the random walk as implemented in Ref.~\cite{FR16}, using directly the function~\eqref{eq:naive-fig-of-merit} as figure-of-merit. The random walk was used to sample a total of 32768 data points, using a binning analysis as described in Ref.~\cite{Ambegaokar-Troyer-MH-methods}, with a step size of $\sim 0.001$, a sweep size of $\sim 1000$ and using $2048$ thermalization sweeps. Again, two choices of the jump distribution give similar results.

Second, we run the \emph{channel-space method} of analysis as presented in Appendix~\ref{sec:main}.  The random walk is run on the space of all quantum processs, as described in Appendix~\ref{sec:numerics}, until 32768 data points have been collected.  Two ways of performing the random walk ($e^{iH}$ versus elementary rotation) yield similar results, with elementary rotation finishing faster than $e^{iH}$. Samples from the random walk allow to construct a numerical estimate of a specific distribution of the figure-of-merit, which contains all the necessary information in order to construct confidence regions.

\begin{figure*}
  \centering
  \includegraphics[width=2\columnwidth]{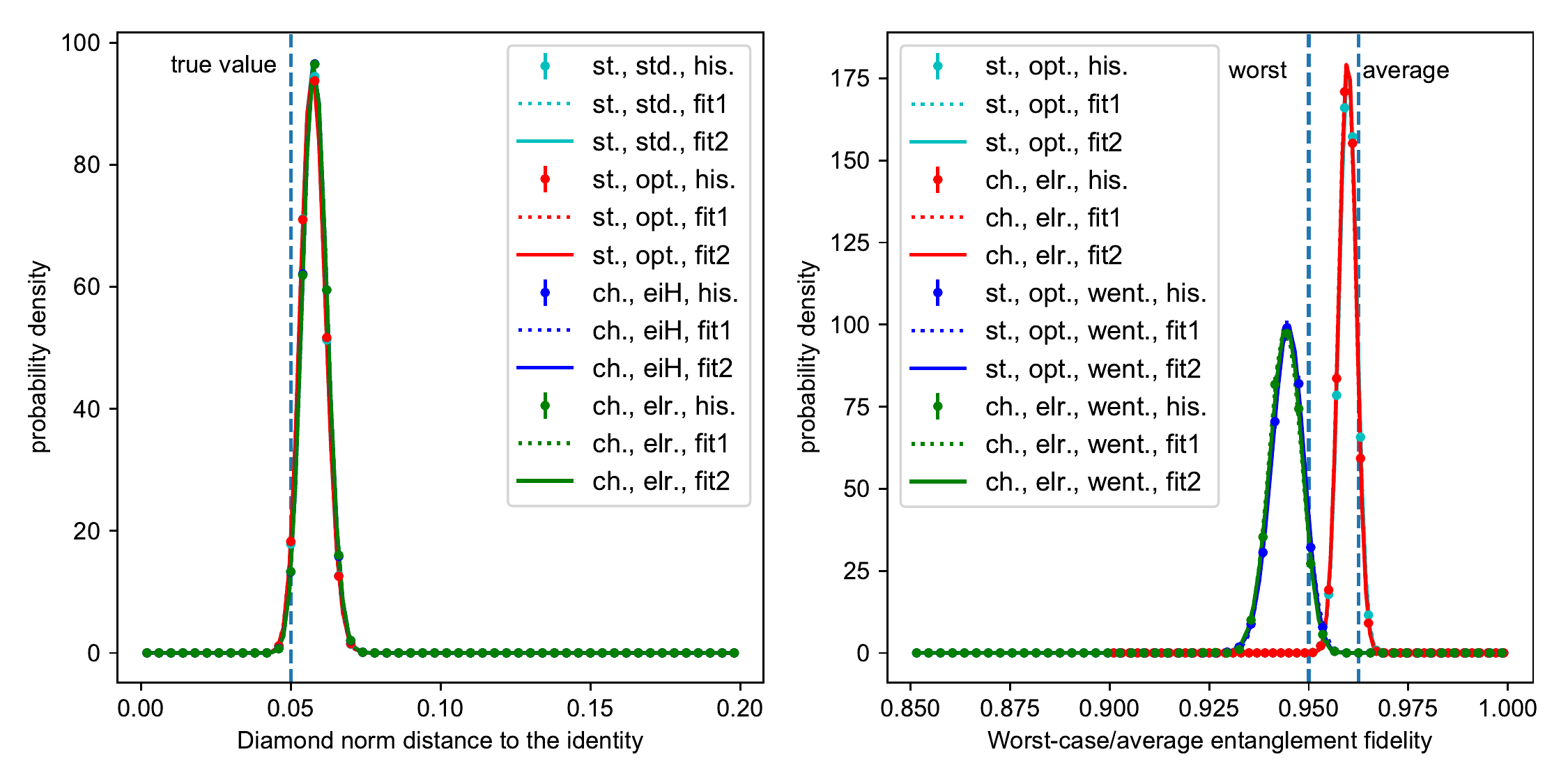}
  \caption{Distribution of the figures-of-merit for the single-qubit process example (left: diamond norm distance, right: entanglement fidelities) relevant to construct confidence regions. Vertical dashed lines are true figure-of-merit values of the unknown channel. The biparite-state sampling method (legend st., shorthand for \emph{state}) consists in estimating the diamond norm while ignoring the information about the exact input state to the channel.  In contrast, the channel-space method (legends ch., shorthand for \emph{channel}) uses the information about the input state to obtain better bounds on the figures-of-merit.  Within each method, we also plot the results obtain from different jump distributions (legends eiH, elr. for channel-space, and std., opt. for bipartite-state) used in the Metropolis-Hastings random walk. The dotted curves are the fits of the raw histogram bins (legend his.) according to our fit model \#1, with corresponding \emph{quantum error bars} $(v_0,\Delta,\gamma)$~\protect\cite{FR16}, while the solid curves are fits using our improved, empirical fit model \#2. These plots should be understood as tools to construct confidence regions, i.e., given a threshold on the $x$-axis, one may easily calculate from these curves the confidence with which one may ascertain the true figure-of-merit (see main text).}
  \label{fig:results-example-plots}
\end{figure*}

The results are shown in \autoref{fig:results-example-plots} as the histogram (dot) points with legend label ``his.''.  The histogram points correspond to the numerical estimation of $h(v)$ given by~\eqref{eq:channel-space-h-of-v} and $\mu(v)$ given by~\eqref{eq:method-naive-mu-of-f}.\\

\emph{Step 4: Fit analysis of histograms:}

In each of these methods, the data---the points underlying the histograms---is fit to two different models as discussed.  If good fit is achieved, we can take these models as a description of the histogram points, and therefore also a description of the functions $h(v)$ and $\mu(v)$. 

In our example, we discovered that the fit model \#1 as described in~\autoref{eq:fit-model-1} does not have great agreement with the underlying histogram bins, as underscored by goodness-of-fit values (reduced $\chi^2$) of the order of $\sim 25$.  This is because our (diamond distance, worst-case entanglement fidelity) figure-of-merit does not satisfy the requirements of the ``heuristic derivation'' in Ref.~\cite{FR16}, and it is thus no surprise that the fit model does not align perfectly well with the data. Using the empirical model \#2 yields much better agreement (solid curves in \autoref{fig:results-example-plots}), with goodness-of-fit values (reduced $\chi^2$) of $\sim 2$. Nevertheless, we reported quantum error bars using fit model \#1.\\

\emph{Step 5. Quantum error bars and confidence regions:}

The quantum error bars $(v_0, \Delta, \gamma)$ are a simple translation from the parameters of the fit model \#1. The steps towards a confidence region for diamond norm has been illustrated in~\autoref{sec:workflow}. In theory we have the guarantee that collecting a larger dataset will yield smaller regions converging to the true value. Unfortunately, the confidence interval for diamond norm distance returned by our method is unreasonably large for the current example: for $99\%$ confidence level we are able to bound the diamond norm by $0.34$ as compared to the true value of $0.05$.  We believe that this is due to operator inequality involved in bounding the failure probability (Proposition~\ref{pro:operatorinequality}). Further research is needed to provide better construction of confidence regions (i.e.\@ more efficient in terms of the number of data samples $n$).

\subsection{Two-qubits example}
Now we consider a two-qubit example to illustrate the practicality of our method in this situation.  This example also shows that the channel-space and the biparite-state sampling methods do not in general produce the same histogram.

Suppose that the real channel implemented in the experiment the two-qubits depolarizing channel
\begin{align}
  \Lambda_{A\to B}(\rho) = p\,\rho + (1-p)\,d_B^{-1}\id_B\ ,
\end{align}
with $d_A=d_B=4$ and we are interested in the diamond distance to the identity channel. Assuming access to state preparation that produces $\ket\psi_{AP} = (\sigma_A^{1/2}\otimes\Ident)\, d_A^{1/2}\ket{\hat\Phi}_{AP}$ with
\begin{align}
  \sigma_A = 
  \begin{pmatrix}
    0.35 & 0 & 0.04 & 0.1i \\
    0 & 0.15 & 0.05 & 0 \\
    0.04 & 0.05 & 0.32 & 0 \\
    -0.1i & 0 &  0 & 0.18
  \end{pmatrix}
  \, ,
\end{align}
and $3^4=81$ Pauli measurement settings each having $2^2=4$ outcomes. We perform similar analyses on a simulated dataset of size $n=40500$  which we generated using the state preparations and measurements described above. The result is presented in~\autoref{fig:twoqubits}.

\begin{figure}
  \centering
  \includegraphics[width=\columnwidth]{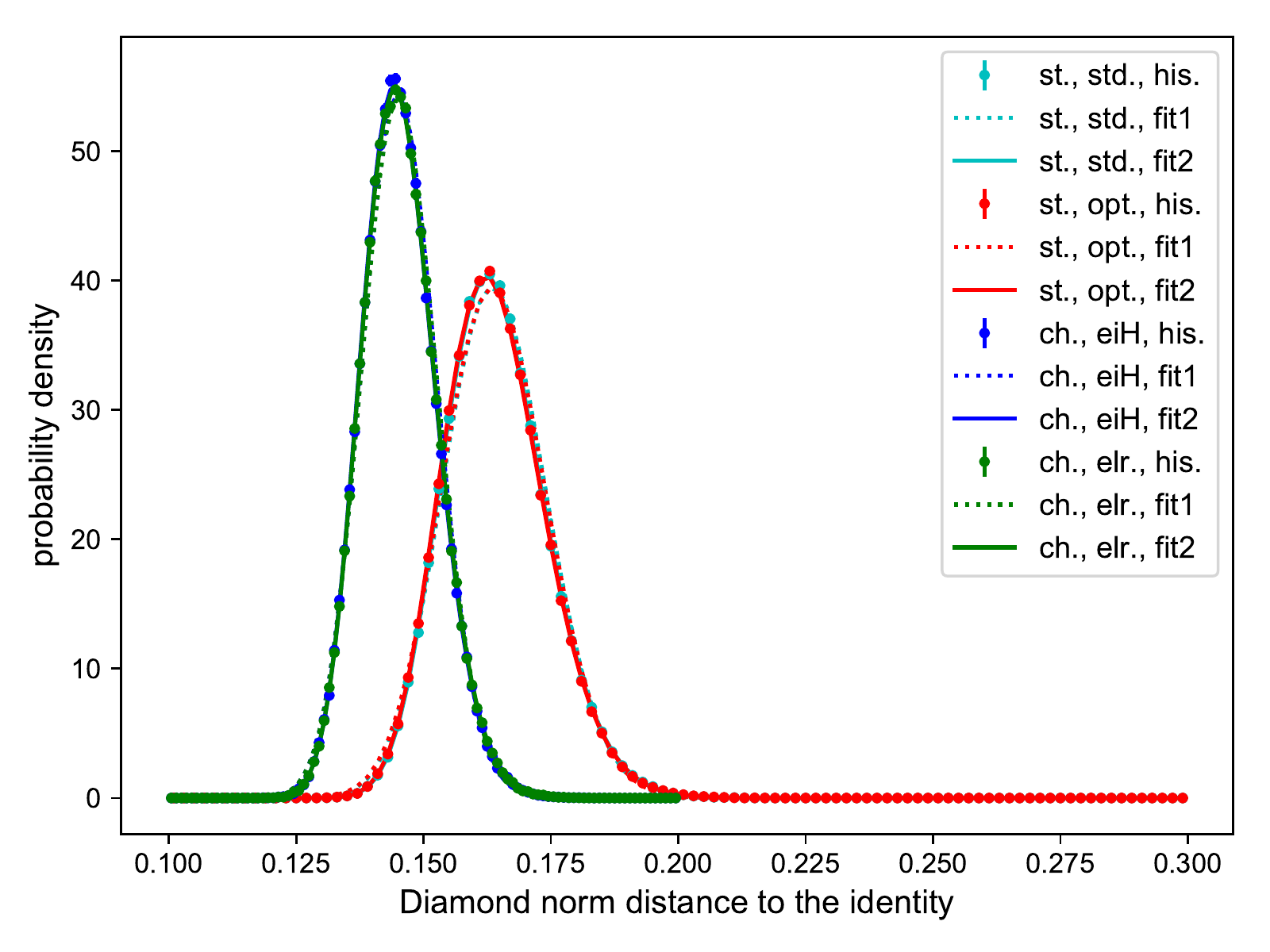}
  \caption{Distribution of the diamond norm distance for the two-qubit process example.  The difference between the two methods for the entanglement fidelity is not a contradiction, rather, in this situation the channel space method gives better tomographic results compared to the biparite-state method. The reason is due to the additional use of prior information about the input state in the channel space method.}
  \label{fig:twoqubits}
\end{figure}

The channel-space sampling method's $h(v)$ is peaked at lower values of the figure-of-merit, as can be seen in \autoref{fig:twoqubits}.  We observe that, in this case, the knowledge of the input state significantly shifts the corresponding histogram distribution towards lower values of the figure-of-merit, allowing to construct smaller confidence regions.  Based on several examples studied, this is not always the case; with less noise (smaller $p$), for instance, the curve for $\mu(v)$ and the curve for $h(v)$ get closer to each other.

On a technical level, we show that the Hilbert-Schmidt measure over the
  bipartite states factorizes as a measure over states on the input system and
  the relevant measure over all channels (Appendix~\ref{sec:compare}).  Hence, a
  large uncertainty over the input state may enlarge the resulting region as
  opposed to considering a region only on the channel space for a fixed known
  input state.  However, it is not impossible that under some lucky
  circumstances a finite distribution width on the input state helps add more
  weight to regions of a higher figure of merit, effectively shrinking the
  region.  Indeed, it could happen that the input state assumed in the
  channel-space method is far from the optimal state for distinguishing the
  channels in terms of the diamond norm; in such a case a prior which is more
  ``smeared out'' over different input states might result in smaller quantum
  error bars for the diamond norm.  We believe that this is why neither method
  performs globally better than the other.  See Appendix~\ref{sec:compare} for further details on the relationship between
the two methods.

\section{\label{sec:conclusion}Conclusions}

  One might think that carrying over the notion of quantum error bars in quantum
  state tomography to quantum process tomography is as straightforward as
  converting quantum states to channels via the Choi-Jamio\l{}kowski
  isomorphism.
  However, our study reveals a more complicated structure.  We find that
  different analysis methods are suited to different experimental process
  tomography setups.  In the experimentally more realistic prepare-and-measure
  scheme, a judicious use of the prior knowledge about the input state to the
  process allows us in typical situations to obtain tighter quantum error bars
  for the process.  These results are obtained by developing a new method, along
  with corresponding proofs, which are specific to process tomography.
  On the other hand, in the case of the ancilla-assisted scheme, we can directly
  apply the methods developed for quantum state tomography, harnessing them
  to directly yield reliable statements about the quantum process itself, while
  ignoring any information the measurements provide about the input state used
  to probe the process.%

  We hence provide a fully-fledged and practical toolbox named {\em QPtomographer}, with solid theoretical
  foundations, for quantum process tomography of arbitrary quantum processes,
  using any experimental quantum process tomography setup, and given measurement
  outcomes from any measurement settings.  Our software package facilitates the
  numerical analysis in practice by automating the implementation of the
  Metropolis-Hastings random walk, as well as the calculation of the diamond
  norm, by simple high-level Python function calls, while transparently
  delegating the computation-intensive routines to heavily optimized C++ code
  which makes use of modern programming techniques including template
  metaprogramming and exploiting hardware SIMD instructions.

  On the spectrum of characterization tools for quantum devices, our method
    can be seen as lying on the opposite end of randomized
    benchmarking~\cite{Knill2008PRA_randbench,Chow2009,Kimmel2014,Jonas}.
    While slightly more involved, our technique can be applied to any choice of
    state preparations and measurements, and can be applied to any individual
    process.  By determining the diamond norm or the worst-case entanglement
    fidelity to any given ideal process, we provide individual full
    characterization of the processes implemented by individual gates.  More
    generally our methods allow the reliable estimation of any specific property
    of the quantum process.

  We note that our method is currently limited to processes acting on few qubits, as our confidence region produces unreasonably large regions, and the algorithm stores dense representations of the quantum process.  However, we expect that our methods will be used to certify individual components of complex setups, for instance, individual 2-qubit gates.  Because we estimate robust, composable figures-of-merit such as the worst-case entanglement fidelity or the diamond norm, the composition of individually certified components is still certified to function accurately.
  
Finally, we may ask whether the channel method is always superior to the
bipartite sampling method.  As noted above, the additional prior knowledge
  about the input state which the channel-space method enjoys in contrast to the
  bipartite sampling method is not sufficient to guarantee this.  We leave a
  more precise understanding of the relation between our two methods open for
  future study.

\begin{acknowledgements}
 TLP, JH, DE and SW are supported by an ERC Starting Grant (SW), an NWO VIDI Grant (SW), and an NWO Zwaartekracht grant (QSC). PhF acknowledges support from the Swiss National Science Foundation through the Early PostDoc.Mobility Fellowship No.~{P2EZP2\_165239} hosted by the Institute for Quantum Information and Matter (IQIM) at Caltech, as well as from the National Science Foundation.
\end{acknowledgements}

\onecolumngrid
\appendix

\section{\label{sec:prelim}Notations \& preliminaries}
We begin by setting up some notations and recalling standard definitions. For more information on states and processes see~\cite{Wolf_book,Chuang_book,Watrous_book}.

\subsection{Quantum processes and figures-of-merit}
Let $\cH_A$ be the Hilbert space of dimension $d_A$ associated with the quantum system denoted $A$. By $\rmD(\cH_A)$ we mean the subset of $\End({\cH_A})$---the set of linear transformations on $\cH_A$---consisting of density matrices $\rho_A\geq0$ (positive semidefinite) with $\tr(\rho_A)=1$. Composite systems are described by tensor product constructions, for instance $\cH_{AB}=\cH_A\otimes\cH_B$ is the Hilbert space of composite system $AB$. 

Quantum measurements on quantum system are the positive operator valued measures or POVMs on $\cH$. For finite number of outcomes, a POVM is a set of positive operators---the effects---that sum to the identity operator on $\cH$. We will overload the notation $E$ to mean an outcome label, and also the effect $E$ (i.e.\@ an operator/matrix) in the POVM. This is equivalent to the usual ``observables'' formulation of measurement, i.e.\@ a hermitian operator. For example, a $Z$ measurement/observable has two outcomes $E=+1$ and $E=-1$ with associated effects $\ketbra{0}$ and $\ketbra{1}$, respectively.

A quantum process $\Lambda_{A\rightarrow B}$ mapping a quantum system $A$ to a quantum system $B$ is a completely positive trace-preserving linear map from $\End{(\cH_A)}$ to $\End{(\cH_B)}$. In general we will denote quantum processs by capital greek letters. We will often drop the subscripts when the quantum systems are clear from the context.

The set of all possible quantum processs is denoted $\ChannelH$, and it is in one-one correspondence with the set of bipartite Choi states $\ChoiH$ via the Choi-Jamiolkowski isomorphism
\begin{align}
J:\ChannelH&\rightarrow\End{(\cH_A\otimes\cH_B)}\\
\Lambda_{A\rightarrow B} &\mapsto (I_A\otimes\Lambda_{\bar A\rightarrow B})(\ketbra{\hat\Phi}_{A\bar A})\notag
\end{align}
where $\ket{\hat\Phi} := \frac{1}{\sqrt{d_A}}\sum_k \ket k_{A}\ket k_{\bar A}$ is the maximally entangled state on $\cH_{A}\otimes\cH_{\bar A}$ and $I_A$ is the identity channel acting on the system $A$. Explicitly, the set of Choi matrices is defined as the image of $\ChannelH$ under the Choi-Jamiolkowski isomorphism and has the following compact description
\begin{align}
\ChoiH = \left\{\rho\in\rmD(\cH_{AB}):\tr_{B}(\rho_{AB})=\id_A/d_A
\right\}\,.\label{eq:Choistates}
\end{align}
Throughout the appendix, we will use the \emph{convention} that $\Lambda_{AB}$ is the Choi state associated with the channel $\Lambda_{A\to B}$.

The action of the channel can be recovered from its Choi state by the inverse of Choi-Jamiolkowski isomorphism
\begin{align}
    \Lambda(\rho) &= d_A\tr_A(\Lambda_{AB} \cdot \rho_{A}^\intercal\otimes\id_{B})\,,\label{eq:channelfromChoi}
\end{align}
where $\intercal$ is the transpose with respect to the basis of $\cH_A$ defining the maximally entangled state.

Recall that the fidelity between two states $\sigma,\sigma'$ is defined as
\begin{align}
F(\sigma,\sigma'):=\norm[\big]{\sqrt{\sigma}\sqrt{\sigma'}}_1=\tr\sqrt{\sqrt{\sigma}\sigma'\sqrt{\sigma}}\,,
\end{align}
and the purified distance between quantum states is defined as $P(\sigma,\sigma'):=\sqrt{1-F(\sigma,\sigma')^2}$. Then the purified distance between channels is defined as
\begin{align}
P(\Psi_{A\to B}, \Psi'_{A\to B}) := P(\Psi_{AB},\Psi'_{AB})\,.
\end{align}

\subsubsection{Diamond distance}
We first introduce the familiar diamond distance. The diamond distance from the real or actual implementation $\Lambda_{A\to B}$ to the ideal or target implementation $\Lambda_\mathrm{ideal}$ is denoted as
\begin{align}
    f_\diamond(\Lambda_{A\to B}) = \frac{1}{2}\norm{\Lambda_{A\to B}-\Lambda^\mathrm{ideal}_{A\to B}}_\diamond.
\end{align}
This function $f_\diamond:\ChannelH\rightarrow[0,1]$ (or equivalently from $\ChoiH$) can be cast as a semidefinite program~\cite{W09}
\begin{center}
    \centerline{\underline{Primal problem}}\vspace{-5mm}
    \begin{align*}
      \text{maximize:}\quad & \ip{d_A \Lambda_{AB} - d_A \Lambda^\mathrm{ideal}_{AB}}{W}\\
      \text{subject to:}\quad & W\leq \id_\cB\otimes \rho,\\
      & W\geq0,\\
      & \rho\in\rmD(\cX).
    \end{align*}
    \centerline{\underline{Dual problem}}\vspace{-2mm}
    \begin{align*}
      \text{minimize:}\quad & \norm{\tr_{B}(Z)}_{\infty}\\
      \text{subject to:}\quad & Z\geq d_A \Lambda_{AB} - d_A \Lambda^\mathrm{ideal}_{AB},\\
      & Z\geq0.\\
    \end{align*}
\end{center}

\subsubsection{Entanglement fidelity}
The entanglement fidelity is another measure of how close a given channel is to the identity channel.  More specifically, it measures how well a channel preserves the maximally entangled state.

  The \emph{entanglement fidelity} of a channel $\Lambda_{A\to B}$ with $B\simeq A$ is defined as
  \begin{align}
    F_e(\Lambda)  = F^2(\Lambda_{\bar{A}\to B}(\hat\Phi_{A\bar{A}}), \hat\Phi_{AB})\ ,
  \end{align}
  recalling that $\ket{\hat\Phi}_{A\bar{A}}$ is the normalized maximally entangled state between the systems $A$ and $\bar A$.

  Because $\Lambda_{AB} = \Lambda_{\bar{A}\to B}(\hat\Phi_{A\bar{A}})$ is the normalized Choi state corresponding to the channel $\Lambda_{A\to B}$, the entanglement fidelity of the channel $\Lambda_{A\to B}$ is in fact exactly the fidelity of the corresponding normalized Choi state to the maximally entangled state:
  \begin{align}
    F_e(\Lambda) = F^2(\Lambda_{AB}, \hat\Phi_{AB})\ .
  \end{align}

\subsubsection{\label{sec:worst-case-ent-fidelity}Worst-case entanglement fidelity}
  The worst-case entanglement fidelity is a better measure of the reliability of the channel to simulate the identity channel, if we have to worry about any possible input state being fed into the channel.  In effect, the worst-case entanglement fidelity measures how well the channel preserves any given state on a system and any purification.  It is defined as
  \begin{align}
    F_{\mathrm{worst}}(\Lambda_{A\to B}) =
     \inf_{\sigma_{A\bar{A}}} F^2(\Lambda_{\bar{A}\to B}(\sigma_{A\bar{A}}), \sigma_{AB})\ ,
  \end{align}
  where the optimization ranges over all bipartite quantum states $\sigma_{AB}$ defined over the input $\bar{A}$ and a reference system $A\simeq \bar{A}$.  The optimization variable, which appears in both slots of the fidelity $F$, may be restricted to pure states without loss of generality.

  Now we show that the worst-case entanglement fidelity can be computed by evaluating a simple semidefinite program.  That a semidefinite program formulation of the worst-case entanglement fidelity can be used in the context of quantum error correction to find suitable recovery procedures for fixed input were put forth in refs.~\cite{Yamamoto2005PRA_suboptimal,Fletcher2007PRA_optimum}.  We build upon those constructions to optimize over the input state, while in our case the problem is simplified as there is no recovery operation.  Using our notation, we write
  \begin{align}
    F_{\mathrm{worst}}(\Lambda_{A\to B}) =
    \inf_{\ket\phi_{A\bar{A}}} F^2(\Lambda_{\bar{A}\to B}(\phi_{A\bar{A}}), \phi_{AB})
    = 
    \inf_{T_{A}:\;\tr(TT^\dagger)=1} \tr(\Lambda_{\bar{A}\to B}(T_A\tilde\Phi_{A\bar{A}} T_A^\dagger) \; T_A\,\tilde\Phi_{AB} T_A^\dagger)\ ,
    \label{eq:wfhusabdhjnfk}
  \end{align}
  where we have defined the non-normalized maximally entangled state $\ket{\tilde\Phi}_{AB} = d_A^{1/2}\,\ket{\hat\Phi}_{AB}$. The last equality comes from the fact that any bipartite pure state $\ket{\phi}_{A\bar{A}}$ can be parametrized by a complex matrix $T_A$ satisfying $\tr(T_AT_A^\dagger)=1$ via $\ket{\phi}_{A\bar{A}} = T_A\,\ket{\tilde\Phi}_{A\bar{A}}$, with moreover $\tr_{\bar A}(\phi_{A\bar A}) = T_A T_A^\dagger$ (indeed, choose $T_A$ with matrix elements $\matrixel{i}{T}{j}_A = \braket{i,j}{\phi}$).  Then, with $T'_A := T_A^\dagger$, and noting that all $T'_A T_A'^\dagger$ with $\tr(T'_A T_A'^\dagger)=1$ can be written as a density matrix $\rho_A = T'_A T_A'^\dagger$, we have
  \begin{align}
    \text{\eqref{eq:wfhusabdhjnfk}}
    &= \inf_{T'_A:\;\tr(T'^\dagger T')=1} \tr(T'_A\,T_A'^\dagger\, \Lambda_{\bar{A}\to B}(\tilde\Phi_{A\bar{A}}) \, T'_A\,T_A'^\dagger\,\tilde\Phi_{AB})
      \nonumber\\
    &= \inf_{\rho_A\geqslant 0:\;\tr(\rho_A)=1}
      \bra{\tilde\Phi}_{AB}\,\rho_A\, \Lambda_{\bar{A}\to B}(\tilde\Phi_{A\bar{A}}) \, \rho_A\,\ket{\tilde\Phi}_{AB}\ .
  \end{align}
  This is a minimization over a positive semidefinite quadratic form in $\rho_A\,\ket{\tilde\Phi}_{AB}$, so it is (quite surprisingly) a convex optimization in terms of $\rho_A$.  We know that positive semidefinite quadratic optimizations may be written as semidefinite programs.  Indeed, for any positive semidefinite matrix $Q = MM^\dagger$, we have that $\matrixel{\psi}{Q}{\psi}\leqslant \mu$ if and only if $\left[\begin{smallmatrix} \Ident & M^\dagger\ket\psi \\ \bra\psi\,M & \mu \end{smallmatrix}\right] \geqslant 0$.  So, finally, we may write the worst-case entanglement fidelity as a semidefinite program in terms of the real variable $\mu$ and the positive semidefinite variable $\rho_A \geqslant 0$:
\begin{align}
   \begin{array}{rcl}
   F_{\mathrm{worst}}(\Lambda_{A\to B})\quad =
   \quad\mbox{minimize:}\quad
   & \mu \qquad            &\ , \\[1ex]
   \mbox{subject to:}\quad
   & \tr(\rho_A) = 1 &\\[1ex]
   & 
     \left[\begin{array}{c|c}
      \id\vphantom{\Bigg[]} & M_{AB}^\dagger \rho_A\,\ket{\tilde\Phi}_{AB} \\ \hline
             \bra{\tilde\Phi}_{AB}\,\rho_A M_{AB} & \mu
     \end{array}\right] \geqslant 0 &
   \end{array}
\end{align}
where $M_{AB}$ is a factorization of the nonnormalized Choi matrix of the process, satisfying
\begin{align}
  M_{AB} \, M_{AB}^\dagger = d_A \Lambda_{AB}
  = \Lambda_{\bar{A}\to B}(\tilde\Phi_{A\bar{A}})\ .
\end{align}
The factorization can be obtained using a Cholseky or LDLT factorization, for instance; or more generally by computing any matrix square root.  The unitary freedom of the matrix square root decomposition (i.e., the freedom of redefining $M\to MU$) is irrelevant here.

\subsection{\label{sec:measures}Haar induced measures}
Later, we will base our confidence region estimators on the following two ``uniform'' measures. They are both measures induced by the unique Haar measure on the unitary group $\bU(\cH)$ acting on some Hilbert space.

The first measure is defined on the set of mixed quantum states~\cite{Zyczkowski2001_Induced}. Since any density matrix has a (nonunique) purification, the space $\rmD(\cH_{AB})$ admits a purification space $\mathrm{Pure}(\cH_{ABA'B'})$ whose elements are rank one density operators on $\cH_{ABA'B'}$ with $A'B'$ being an isomorphic copy of $AB$. The Haar measure $dU_{ABA'B'}$ then induces a measure on $\mathrm{Pure}(\cH_{ABA'B'})$ via the relation $\ketbra{\psi}=U\ketbra{\psi_0}U^\dagger$ for an arbitrary pure state $\ket{\psi_0}$, which induces a measure $\diff\sigma_{AB}$ on $\rmD(\cH_{AB})$ by partial tracing.

The second measure is defined on the set of quantum processs, or equivalently on the set of bipartite Choi states. Let
\begin{align}
    \sP\sC\! =\! \left\{\ket{\Psi}\in\cH_{ABA'B'}:\tr_{BA'B'}(\ketbra{\Psi})=d_A^{-1}\id_A\right\}
\end{align}
be the set of purifications of arbitrary Choi states. Without loss of generality, let us define a fixed reference pure state in $\sP\sC$
\begin{align}
\ket{\Psi_0}:=\frac{1}{d_A}\sum_{i=1}^{d_A}\ket{i}_A\ket{v_i}_{BA'B'}\label{eq:fixedpurifiedChoi}
\end{align}
with $\{\ket{v_i}_{BA'B'}\}$ some fixed orthonormal set of vectors. Then for all $\ket{\Psi}\in\sP\sC$, there exists a unitary $U_{BA'B'}$ such that $\ket{\Psi}=\id_A\otimes U_{BA'B'}\ket{\Psi_0}$. This relation transfer the unique Haar measure $\diff U_{BA'B'}$ on the unitary group $\bU(\cH_{BA'B'})$ to a measure on $\sP\sC$ which we will denote as $\diff \nu(\ket{\Psi})$. Again, by partial tracing the system $A'B'$, this measure induces the measure $\diff \nu(\Psi_{AB})$ on Choi states $\ChoiH$ (also denoted as $\diff \nu(\Lambda_{AB})$ by changing the dummy variable). Finally, taking the inverse of the Choi-Jamiolkowski isomorphism gives the induced measure $\diff \nu(\Psi_{A\to B})$ (or in a different notation $\diff \nu(\Lambda_{A\to B})$) on channel space $\ChannelH$ which is the starting point of the channel space sampling method.

The relation between these measures will be discussed in Appendix~\ref{sec:compare} when we compare the two region estimators.

\subsection{The i.i.d. hypothesis}
In this paper,~\emph{we work under the assumption of i.i.d. (independent and identically distributed) channels}. This means any time we use the experimental device, it is assumed that one and the same transformation $\Lambda_{A\to B}$ has been applied. Experimentally, this assumption is well justified if the same experimental conditions can be reproduced because the abstract channel is a function of the working parameters of the physical device. The i.i.d. hypothesis also gives a clear operational meaning to the question: to which object does the tomographic statement apply? It is one and the same $\Lambda_{A\to B}$ which does not vary from past to future uses.

Even though we work under the i.i.d. assumption, we note that this can be weakened to permutation invariant through the use of the quantum de Finetti theorem for channels~\cite{deFinettiChannels}.

Before proceeding further, we give a clarifying remark about our notation. We usually consider $n$ uses of a channel $\Lambda$. Under the i.i.d. assumption we can describe this situation by tensor product construction giving a composite channel $\Lambda^{\otimes n}$ acting on the composite Hilbert space $\cH_A^{\otimes n}$ and transform the system to $\cH_B^{\otimes n}$. As usual, by measurement on $\cH_B^{\otimes n}$ and by knowing the input state on $\cH_A^{\otimes n}$ we can perform tomography of the unknown channel. Our convention has been to denote a measurement on $\cH_B^{\otimes n}$ by a POVM $\{E\}$ with $E$ standing for both the labels of the various outcomes and the actual operators/matrices. This captures both i.i.d. measurements and entangled measurements in the following sense. Suppose $n=2$ and we perform $X$ and $Z$ on each subsystem. This can be equivalently described by two POVMs $\{\ketbra{+x},\ketbra{-x}\}$ and $\{\ketbra{+z},\ketbra{-z}\}$, and then by tensor product construction combined into a single POVM on the composite Hilbert space. However, this is not the only measurement that one can do: one can perform the Bell measurement projecting into the four maximally entangled states. Our description and notation is flexible for arbitrary measurement one can perform.

\section{\label{sec:naive}The bipartite-state sampling method}

This method requires experimentalists to work in the ancilla-assisted scheme (see~\autoref{fig:tomo}\textbf{(b)}): we select a full Schmidt rank entangled state $\psi_{AP}$, a collection of bipartite measurements $E^{(\ell)}$ with corresponding effects $E^{(\ell)}_k$, and assume the experiment can implement the channel $\Lambda\otimes\mathcal{I}$, where $\mathcal{I}$ is the identity map. Again, the collection of measurement should be informationally complete if one wishes to infer full information about the channel. We assume knowledge of the state preparations and measurements in the form of matrices in the computational basis. This means the pure entangled state has the form
\begin{align}
\ket\psi_{AP} = \sum_is_i\ket{i}_A\ket{i}_P = \sqrt{d_A}\psi_A^{1/2}\ket{\hat\Phi}_{AP} = \sqrt{d_A}\psi_P^{1/2}\ket{\hat\Phi}_{AP},
\end{align}
where $\psi_A,\psi_P$ are the respective reduced states on $A$ and $P$ of $\ketbra{\psi_{AP}}$ and $\ket{\hat\Phi}_{AP}$ the maximally entangled state on $\mathcal{H}_{AP}$.  Note that not all pure state on $AP$ has this form, but we assume it without loss of generality by redefining $\ket{\hat{\Phi}}_{AP}$ if necessary.

The tomography procedure proceeds according to Algorithm~\ref{ancilla-assisted}. In each round, we prepare $\ket{\psi}_{AP}$ and we apply the unknown channel $\Lambda_{A\to B}\otimes\mathcal{I}_{P\to P}$.  We then perform a measurement on the bipartite output system $BP$ using a setting of our choice, yielding an outcome POVM effect $E_k^{(\ell)}$. The dataset stores all the outcomes of different rounds. 

In other words, the ancilla-assisted scheme actually realizes the (theoretical) Choi-Jamio\l{}kowski isomorphism in the laboratory under the assumption of the input state and the channel.

The likelihood function for this scheme is given by
\begin{align}
  \mathcal{L}_{\textrm{AA}}(\Lambda|E)
  = \prod_{k,\ell} \left[\tr(\Lambda_{A\to B}(\psi_{AP})\,E_k^{(\ell)})\right]^{n_{k,\ell}}\ ,
\end{align}
where $n_{k,\ell}$ is the number of times the POVM effect $E_k^{(\ell)}$ appears in the dataset $E$.  Since 
\begin{align}
\Lambda_{A\to B}(\psi_{AP}) = d_A \psi_P^{1/2} \, \Lambda_{BP} \, \psi_P^{1/2}
\end{align}
where $\Lambda_{BP}$ is the corresponding Choi state, we have
\begin{align}
  \mathcal{L}_{\textrm{AA}}(\Lambda|E) = \prod_{k,\ell} \left[d_A\tr(\Lambda_{BP}\,\psi_P^{1/2}\,E_k^{(\ell)}\,\psi_P^{1/2})\right]^{n_{k,\ell}} = d_A^n\tr\bigg(\Lambda_{BP}^{\otimes n}\,\bigotimes_{k,\ell} \psi_P^{1/2}\,E_k^{(\ell)}\,\psi_P^{1/2}\bigg)\ ,
\end{align}
where $\bigotimes_{k,\ell}$ ranges over the observed dataset $E$

Since quantum processs correspond to bipartite quantum states via the Choi-Jamiolkowski isomophism, we can generalize quantum state tomography methods to quantum processs. Here, we \emph{directly} apply the existing procedure of Faist and Renner~\cite{FR16} designed for quantum states to infer information about quantum processs. The main result in this section is Theorem~\ref{thm:bipartite_sampling}. We first recall the procedure of constructing confidence region estimators for quantum states, phrased in terms of bipartite states in anticipation with the connection to quantum processs. 

\subsection{Christandl-Renner confidence regions}
Given access to $n$ copies of an unknown state $\rho_{AB}$, we can perform a (joint or collective) POVM measurement on $\rho_{AB}^{\otimes n}$ and upon receiving the dataset $E$, the Christandl-Renner procedure outputs a distribution
\begin{align}
\diff \mu_E(\sigma_{AB}) := c_E^{-1}\tr(\sigma_{AB}^{\otimes n} E)\diff\sigma_{AB}
  \label{eq:def-mu-E-sigmaAB}
\end{align}
where $c_E=\int\tr(\sigma_{AB}^{\otimes n} E)\diff\sigma_{AB}$ and $\diff\sigma_{AB}$ is the uniform distribution on bipartite density matrices. Confidence regions for the unknown $\rho_{AB}$ can be constructed from $\diff\mu_E(\sigma_{AB})$ as the following proposition asserts.
\begin{theorem*}[\ref{thm:CRtomo} of main text]
Let $n$ be the number of systems measured by a POVM during tomography and $1-\epsilon$ be the desired confidence level. For each effect $E$ in the POVM, let $S_{\mu_E}\subseteq\rmD(\cH_{AB})$ be a set of states such that
\begin{align}
\int_{S_{\mu_E}} \diff\mu_E(\sigma_{AB}) \geq 1-\frac{\epsilon}{2}s_{2n,d^2_{AB}}^{-1}\,,
\end{align}
where $s_{n,d}=\binom{n+d-1}{d-1}\leq(n+1)^{d-1}$ and let $S_{\mu_E}^\delta$ be the enlargement of $S_{\mu_E}$ defined as
\begin{align}
S_{\mu_E}^\delta := \{\sigma_{AB}: \exists\sigma'\in S_{\mu_E} \, with \, P(\sigma,\sigma')\leq\delta\}\,.
\end{align}
Then the mapping $E\mapsto S_{\mu_E}^\delta$ is a confidence region estimator for the unknown $\rho_{AB}$ with confidence level $1-\epsilon$ if
\begin{align}
\delta^2 = \frac{2}{n}\left(\ln\frac{2}{\epsilon}+2\ln s_{2n,d^2_{AB}}\right)\,.
\end{align}
In other words, for any $\rho_{AB}\in\rmD(\cH_{AB})$,
\begin{align}
\Pr_E[\rho_{AB}\in S_{\mu_E}^\delta]\geq1-\epsilon\,,
\end{align}
where the probability is taken over the random dataset $E$ with distribution $\tr(\rho_{AB}^{\otimes n}E)$.
\end{theorem*}

\subsection{Mapping channel tomography to biparite-state tomography}
Consider the ancilla-assisted scheme.  In order to learn what the channel $\Lambda_{A\to B}$ is, we may carry out the experiment as described in Algorithm~\ref{ancilla-assisted}, and use the outcome measurements to perform full tomography on the output state $\rho_{BP}:=\Lambda_{A\to B}(\psi_{AP})$.  We may then ask, what does this tell us about the unknown channel $\Lambda_{A\to B}$?

Observe that if we knew the output state $\rho_{BP}$ \emph{exactly} (limit of infinite data) and assume the input state $\psi_{AP}$ has full rank, then we could read out the true channel: its Choi state is simply given as $\Lambda_{BP} = d_A^{-1}\rho_P^{-1/2}\,\rho_{BP}\,\rho_P^{-1/2}$.  Indeed, we have 
\begin{align}
\rho_{BP} := \Lambda_{A\to B}(\psi_{AP}) = d_A\psi_P^{1/2}\,\Lambda_{A\to B}(\hat\Phi_{AP})\,\psi_P^{1/2} = d_A\rho_P^{1/2}\,\Lambda_{BP}\,\rho_P^{1/2} \label{eq:observation}
\end{align}
since under the assumption that $P$ has undergone identity transformation it follows $\psi_P=\rho_P$. Note that this is the same trick used in Appendix~\ref{sec:worst-case-ent-fidelity} to derive the semidefinite program.

Thanks to this observation, we may use the quantum state tomography method of Ref.~\cite{FR16} to construct confidence regions on the space of quantum processs $\Lambda_{A\to B}$, as well as on a figure-of-merit such as the diamond norm to an ideal channel.

To do so, we ignore the knowledge of the exact input state $\psi_A$, but we assume the global state $\ket{\psi_{AP}}$ has full Schmidt rank (i.e.\@ forgetting the Schmidt coefficients). Upon observing the dataset $E$, the classical data processing returns a bipartite state region $S_{\mu_E}^\delta$, which contains information about the pair $(\Lambda_{A\to B}, \psi_A)$~\footnote{In fact, in addition to reconstructing the channel $\Lambda_{A\to B}$, we may use this procedure to confirm the correct preparation of the input state $\sigma_A$.}. The interpretation of $S_{\mu_E}^\delta$ is given by Theorem~\ref{thm:CRtomo}, and together with the observation above (see Eq.~\eqref{eq:observation}) we have
\begin{align}
\Pr_E\left[d_A\rho_{P}^{1/2}\Lambda_{BP}\rho_{P}^{1/2}\in S_{\mu_E}^\delta\right]
  \geq 1-\epsilon\,,
\end{align}
where the probability is taken over all possible dataset $E$ with distribution $\tr((\rho_{P}^{1/2}\Lambda_{BP}\rho_{P}^{1/2})^{\otimes n}E)$. To recover information about the channel $\Lambda$, for each $\rho_{BP}\in S_{\mu_E}^\delta$ we apply the (completely positive) transformation $T$ defined as
\begin{align}
T:~\End(\cH_{AB})&\rightarrow\End(\cH_{AB}) \nonumber \\
\rho_{BP} &\mapsto d_A^{-1}\rho_{P}^{-1/2}\rho_{BP}\rho_{P}^{-1/2}.
\end{align}
Observe that $T$ maps any $\rho_{BP}$ with full rank marginal $\rho_P$ to a Choi state. Also, the set $S_{\mu_E}^\delta$ only contains $\rho_{BP}$ with full rank marginal $\rho_P$ because we only sample according to the uniform measure $\diff\sigma_{AB}$ (i.e.\@ the set of rank-deficient $\rho_{BP}$ has measure zero). This means the image of $S_{\mu_E}^\delta$  under $T$ will be a set of Choi matrices which can be interpreted via Choi-Jamiolkowski as a region of quantum processs (completely positive and trace-preserving maps). We conclude
\begin{align}
\Pr_E\left[\Lambda_{BP}\in T(S_{\mu_E}^\delta)\right]
  \geq 1-\epsilon\,,\label{eq:bipartiteversion-conf}
\end{align}
which implies $T(S_{\mu_E}^\delta)$ are confidence regions for quantum processs.

\subsection{Regions for figures-of-merit}
The confidence region on channel space constructed in the last section contains full information on the unknown channel. But if one is only interested in a property of the channel, for instance how close is it to an ideal process, then obtaining confidence region for a given figure-of-merit suffices. We now present how one can do this using pushforward of measures.

Given a figure-of-merit for quantum processs $f_\mathrm{channel}$ (defined on channel space), we associate a function $f$ defined on the set of bipartite states as
\begin{align}
f(\rho_{BP}) := f_\mathrm{channel}(J^{-1}(d_A^{-1}\rho_P^{-1/2}\,\rho_{BP}\,\rho_P^{-1/2})),
\label{eq:naive-fig-of-merit}
\end{align}
which is just $f_\mathrm{channel}$ acting on the channel $J^{-1}(d_A^{-1}\rho_P^{-1/2}\,\rho_{BP}\,\rho_P^{-1/2})$ obtained from $\rho_{BP}$ via the mapping $T$. This allows us to directly use the tools of Ref.~\cite{FR16} to obtain confidence intervals for the figure-of-merit $f$ which will yield the same result as $f_\mathrm{channel}$. Explicitly, for any $v\in\mathbb{R}$
  \begin{align}
    \mu(v) = \int \diff \mu_E(\sigma_{AB})\,\delta(f(\sigma_{AB}) - v)\,
    \label{eq:method-naive-mu-of-f}
  \end{align}
is the probability density of the pushforward of $\diff\mu_E(\sigma_{AB})$ along $f$. This density provides confidence region for a figure-of-merit as certified by the following proposition.
\begin{theorem}
\label{thm:bipartite_sampling}
  Let $\mu_{E}$ be given as in~\eqref{eq:def-mu-E-sigmaAB}, and let $\mu(v)$ be defined as in~\eqref{eq:method-naive-mu-of-f}.
  Then for any threshold value $v_\mathrm{thres}>0$, the region
  \begin{align}
    R^{v_\mathrm{thres},\delta} = \{
    \rho_{AB} : f(\rho_{AB}) \leqslant v_\mathrm{thres} + O(\delta)
    \}
  \end{align}
  of states representing channels at least $v_\mathrm{thres}+O(\delta)$-close to the reference channel, is a confidence region of confidence level $1-\epsilon$ where
  \begin{align}
    \epsilon = \operatorname{poly}(n)\,\left[
    1 - \int_0^{v_\mathrm{thres}} \mu(v)\,\diff v
    \right]\ .
  \end{align}
\end{theorem}

In summary, for ancilla-assisted tomography scheme, determining the histogram $\mu(v)$ in~\eqref{eq:method-naive-mu-of-f} gives us all the necessary information to construct confidence regions of any confidence level in terms of the figure-of-merit $f_\mathrm{channel}(\Lambda_{A\to B})$.\\

{\bf Diamond distance to ideal and worst-case entanglement fidelity:}

The methodology outlined in the previous paragraphs can be specialized to the diamond distance to an ideal reference channel $\Lambda^\mathrm{ideal}_{B\to P}$. Here we take
\begin{align}
  f_\mathrm{channel}(\Lambda_{B\to P}) = f_\diamond(\Lambda_{B\to P}) = \frac{1}{2}\norm{\Lambda_{B\to P}-\Lambda_{B\to P}^\mathrm{ideal}}_\diamond
\end{align}
to be the desired figure-of-merit on channel space. This induces a figure-of-merit in the space of bipartite quantum states
\begin{align}
  f(\rho_{BP})
  &= \frac{1}{2}\max\left\{\ip{\rho_P^{-1/2}\rho_{BP}\rho_{P}^{-1/2}-d_A\Lambda^{\mathrm{ideal}}_{BP}}{W}:W\leq \id_\cB\otimes \bar{\rho}, W\geq0, \bar{\rho}\in\rmD(\cX)\right\}\,.
\end{align}
One is left to perform a numerical computation of $\mu(v)$ for the above function $f$, as explained in details in Ref.~\cite{FR16}.

\section{\label{sec:main} The channel-space sampling method}
This method applies to either the ancilla-assisted scheme explained in the previous Appendix, or the prepare-and-measure scheme where no entanglement is required. In the prepare-measure scheme (see~\autoref{fig:tomo}\textbf{(a)}), we select a collection of input states $\sigma^j$, and select a collection of measurements $E^{(\ell)}=\{ E^{(\ell)}_k \}$. This set of state preparation and measurement (SPAM) should be informationally complete if one wish to fully reconstruct the unknown channel. The SPAM is represented as certain set of matrices in the computational basis $\{\ket{i}: i=0,...,d_A-1\}$.

The data collection procedure goes as follows: in each round, we choose an input state $\sigma^j$, we choose a measurement $\ell$ on output, we send $\sigma^j$ through the channel, and record the measurement outcome $k$ on the output.  The dataset $E$ consists of all pairs $(\sigma^j,E^{(\ell)}_k)$ chosen and observed for each round.

Typically one can choose the states $j$ in order, i.e., first perform measurements on $\sigma^1$, then on $\sigma^2$, etc.  The choice of the output measurement setting $\ell$ is allowed to depend on $j$.  Since we are under i.i.d. channel assumption, at each round it is the same unknown channel $\Lambda$ which is applied, and that previous outcomes have no influence on new rounds.

The likelihood function for a dataset $E$ in this scenario is defined using the matrix representations of the SPAM according to Born's rule
\begin{align}
    \mathcal{L}_{\textrm{PM}}(\Lambda|E) = \prod_{j,k,\ell} \left[\tr(\Lambda(\sigma^{j})\,E^{(\ell)}_{k})\right]^{n_{j,k,\ell}}\ ,
\end{align}
where $n_{j,k,\ell}$ is the number of times the given pair $(\sigma^j, E^{(\ell)}_k)$ appears in the dataset $E$.  Using~\eqref{eq:channelfromChoi}, we rewrite the likelihood function as
  \begin{align}
    \mathcal{L}_{\textrm{PM}}(\Lambda|E) &= \prod_{j,k,\ell} \left[d_A\tr( \Lambda_{AB} \; (\sigma_A^j)^\intercal \otimes E^{(\ell)}_k )\right]^{n_{j,k,\ell}} =  d_A^n\tr\bigg( \Lambda_{AB}^{\otimes n} \; \bigotimes_{j,k,\ell}(\sigma_A^j)^\intercal \otimes E^{(\ell)}_k \bigg), \label{eq:prob-prep-meas-Choi--sigma-j-Ek}
  \end{align}
where $\bigotimes_{j,k,\ell}$ ranges over the observed dataset $E$.

The method in the previous section maps a channel tomography problem into a (constrained) biparite-state tomography problem. One may ask if this is the only solution. In this section, we provide an alternative construction natively on the channel space. This has consequence on the numerical implementation: we no longer need to samples from biparite-state space. Instead, we can directly sample ``random channels'' which leads to improved numerical efficiency. The main results in this section are Theorems~\ref{thm:channel-space} and~\ref{thm:channel-space-diamond}.

\subsection{Regions on channel space}
Inspired by the Christandl-Renner construction~\cite{CR12}, we define the following confidence region estimator for quantum processs. Our confidence region is constructed from the probability measure on the space of quantum processs $\ChannelH$
\begin{align}
\diff\nu_E(\Lambda):=c'^{-1}_E\mathcal{L}(\Lambda|E)\diff \nu(\Lambda)
\label{eq:def-nu-E-channels}
\end{align}
where $\mathcal{L}(\Lambda|E)$ is either prepare-measure or ancilla assisted likelihood and $c'_E=\int\mathcal{L}(\Lambda|E)\diff \nu(\Lambda)$ serves as a normalizing constant and $\diff \nu(\Lambda)$ is the induced measure on $\ChannelH$ defined in Appendix~\ref{sec:measures}.

The main result in this Section is
 \begin{theorem*}[2 of maintext]
Let $n$ be the number of channel uses during tomography and $1-\epsilon$ be the desired confidence level. For each dataset $E$, let $R_{\nu_E}\subseteq\ChannelH$ be a set of channels such that
\begin{align}
    \int_{R_{\nu_E}}\diff \nu_E(\Lambda) \geq 1-\frac{\epsilon}{2}s_{2n,d^2_{AB}}^{-2}\ ,
\end{align}
where $s_{n,d}=\binom{n+d-1}{d-1}\leq(n+1)^{d-1}$ and let $R_{\nu_E}^\delta$ be the enlargement
\begin{align}
R_{\nu_E}^\delta:=\{\Lambda\in\ChannelH:\exists \Lambda'\in R_{\nu_E} \nonumber \\ 
\textrm{ with } P(\Lambda,\Lambda')\leq\delta\}\,.
\end{align}
Then the mapping $E\mapsto R_{\nu_E}^\delta$ is a confidence region estimator for the unknown $\Lambda_{A\to B}$ with confidence level $1-\epsilon$ if
\begin{align}
\delta^2 = \frac{2}{n}\left(\ln\frac{2}{\epsilon}+3\ln s_{2n,d^2_{AB}}\right)\,.
\end{align}
In other words, for all channel $\Lambda\in\ChannelH$
\begin{align}
       \Pr_E[\Lambda_{A\to B}\in R_{\nu_E}^\delta] \geq 1-\epsilon\,,
\end{align}
where the probability is over the random dataset $E$ with distribution $\Pr(E|\Lambda)=\mathcal{L}(\Lambda|E)$.
\end{theorem*}

Before starting the proof, we will need the following results.

\begin{proposition}
\label{pro:operatorinequality}
For any channel $\Lambda_{A\to B}$, if $\ket{\Lambda}\in\cH_{ABA'B'}$ is a purification of its Choi state then
\begin{align}
\ketbra{\Lambda}^{\otimes n} \leq s^2_{n,d^2_{AB}}\int_{\sP\sC}\diff \nu(\ket{\Psi})  \, \ketbra{\Psi}^{\otimes n} = s^2_{n,d^2_{AB}}\int\diff U \, U_{BA'B'}^{\otimes n}\ketbra{\Psi_0}^{\otimes n} U_{BA'B'}^{\dagger\otimes n}
\end{align}
where $s_{n,d}:=\binom{n+d-1}{d-1}$.
\end{proposition}

\begin{proof}
The main idea of this proof is to discretize the Haar integral using Caratheodory's theorem, and dominate the left hand side by a trivial operator inequality.

By definition, the operator
\begin{align}
\int_{\sP\sC}\diff \nu(\ket{\Psi})  \, \ketbra{\Psi}^{\otimes n}
\end{align}
lies in the convex hull of the set $\{\ketbra{\Psi}^{\otimes n}:\ket{\Psi}\in\sP\sC\}$, whose linear span (in the ambient space $\End(\cH_{ABA'B'}^{\otimes n})$) has dimension $D$. By Caratheodory's theorem, there exists a convex combination $(q_i,\ketbra{\Psi_i}^{\otimes n})$ with size at most $D+1$ such that
\begin{align}
\int_{\sP\sC}\diff \nu(\ket{\Psi})  \, \ketbra{\Psi}^{\otimes n} = \sum_{i=1}^{D+1} q_i \ketbra{\Psi_i}^{\otimes n}\,.
\end{align}
Among the probability weights $q_i$ there exists a largest element denoted $q_\mathrm{max}$ and its associated purified Choi state $\ketbra{\Psi_\mathrm{max}}$, from which we split off this term in the finite sum as
\begin{align}
\sum_{i=1}^{D+1} q_i \ketbra{\Psi_i}^{\otimes n} = q_\mathrm{max}\ketbra{\Psi_\mathrm{max}}^{\otimes n} + \sum_{i\neq\mathrm{max}} q_i \ketbra{\Psi_i}^{\otimes n}\,.
\end{align}
By left-invariance of the measure $\diff \nu(\ket{\Psi})$ and the (unitary) structure of the set $\sP\sC$, we can without loss of generality assume that $\Psi_\mathrm{max}=\Lambda$. More precisely, let $W_{BA'B'}$ be a unitary transformation bringing $\ket{\Psi_\mathrm{max}}$ to $\ket{\Lambda}$, we have (leaving the system label $BA'B'$ implicit)
\begin{align}
W^{\otimes n} \left(\int_{\sP\sC}\diff \nu(\ket{\Psi})  \, \ketbra{\Psi}^{\otimes n}\right) W^{\dagger\otimes n} = q_\mathrm{max}W^{\otimes n}\ketbra{\Psi_\mathrm{max}}^{\otimes n}W^{\dagger\otimes n} + \sum_{i\neq\mathrm{max}} q_i W^{\otimes n}\ketbra{\Psi_i}^{\otimes n}W^{\dagger\otimes n} \,.
\end{align}
Using linearity of integration and translational invariance of the integrating measure, this equation simplifies to
\begin{align}
\int_{\sP\sC}\diff \nu(\ket{\Psi})  \, \ketbra{\Psi}^{\otimes n} = q_\mathrm{max}\ketbra{\Lambda}^{\otimes n} + \sum_{i\neq\mathrm{max}} q_i \ketbra{\Psi_i'}^{\otimes n}\,,
\end{align}
where $\ket{\Psi_i'}$ is some other vector in $\sP\sC$.

Now since all operators in the convex combination are positive-semidefinite, we obtain
\begin{align}
\int_{\sP\sC}\diff \nu(\ket{\Psi})  \, \ketbra{\Psi}^{\otimes n} \geq q_\mathrm{max}\ketbra{\Lambda}^{\otimes n}\,.
\end{align}
By the property of the maximum weight $q_\mathrm{max}$, namely $q_\mathrm{max}\geq1/(D+1)$, we get
\begin{align}
\ketbra{\Lambda}^{\otimes n} \leq (D+1)\int_{\sP\sC}\diff \nu(\ket{\Psi})  \, \ketbra{\Psi}^{\otimes n}\,.
\end{align}

Finally, $\mathrm{span}\{\ketbra{\Psi}^{\otimes n}:\ket{\Psi}\in\sP\sC\} \subseteq \mathrm{span}\{\ketbra{\Psi}^{\otimes n}:\ket{\Psi}\in\cH_{ABA'B'}\}$ and the latter is identified as a subspace of $\End(\mathrm{Sym}^n(\cH_{ABA'B'}))$, the operator space on the symmetric subspace of $\cH_{ABA'B'}^{\otimes n}$. Together with the constraint that trace is $1$, we have $D\leq s_{n,d_{ABA'B'}}^2-1$ where the dimension of the symmetric subspace is $s_{n,d}:=\binom{n+d-1}{d-1}$. This completes the proof of the operator inequality.
\end{proof}

\begin{proof}[Proof of Theorem~\ref{thm:channel-space}] Our proof technique follows closely that of~\cite{CR12}, with the main technical difficulty being incorporating the \emph{a priori} constraint $\tr_B(\Lambda_{AB})=\id_A/d_A$. This allows the reduction of numerical sampling from biparite-state space to channel space.

For any region estimator, our construction $E\mapsto R_{\nu_E}^\delta$ in particular, the failure probability of the reconstruction typically depends on the underlying unknown channel
\begin{align}
P_{\textup{fail}}(\Lambda_{A\to B})=\Pr_E[\Lambda_{A\to B}\notin R_{\nu_E}^\delta] := \sum_E \Pr(E|\Lambda) \chi(\Lambda_{A\to B};\overline{R_{\nu_E}^\delta})\,,
\end{align}
where $\Pr(E|\Lambda)$ is the probability of obtaining dataset $E$, and $\chi(\Lambda_{A\to B};\overline{R_{\nu_E}^\delta})$ is the indicator function of the set $\overline{R_{\nu_E}^\delta}:=\ChannelH\setminus R_{\nu_E}^\delta$ (i.e.\@ the complement set). Recall that
\begin{align}
\Pr(E|\Lambda) = \begin{cases} d_A^n\tr\bigg( \Lambda_{AB}^{\otimes n} \; \bigotimes_{j,k,\ell}(\sigma_A^j)^\intercal \otimes E^{(\ell)}_k \bigg) \text{ in prepare-and-measure scheme} \\
 d_A^n\tr\bigg(\Lambda_{BP}^{\otimes n}\,\bigotimes_{k,\ell} \psi_P^{1/2}\,E_k^{(\ell)}\,\psi_P^{1/2}\bigg) \text{ in ancilla-assisted scheme}
\end{cases}
\end{align}
Our goal will be  bounding this failure probability independently of $\Lambda_{A\to B}$ by using the operator inequality we have just developed.

Before starting the actual calculations, observe that $\Pr(E|\Lambda)$ for both schemes are functions of the type $\tr(\Lambda^{\otimes n}\otimes\cdots)$ where $\otimes\cdots$ is the operator constructed from the observed dataset $E$ from information about the state preparation and measurement schemes. In the following, we do not utilise the exact form of $\otimes\cdots$ for each schemes and thus the calculation works for both schemes. We choose to put $\otimes\cdots$ as the operator corresponding to the prepare-and-measure scheme for concreteness.

Via the Choi-Jamiolkowski isomorphism, the failure probability reads
\begin{align}
P_{\textup{fail}}(\Lambda_{AB})=\Pr_E[\Lambda_{AB}\notin R_{\nu_E}^\delta] := \sum_E d_A^n\tr[\Lambda_{AB}^{\otimes n}\,\rho_{A^n}^\intercal\otimes E_{B^n}] \chi(\Lambda_{AB};\overline{R_{\nu_E}^\delta})\,,\label{eq:Pfail_start}
\end{align}
where we have abused the notation $R_{\nu_E}^\delta$ to mean both the set in channel space $\ChannelH$ and in Choi state space $\ChoiH$. This can be rewritten in terms of an arbitrary purification of the Choi state $\Lambda_{AB}$
\begin{align}
P_{\textup{fail}}(\Lambda_{AB}) = \sum_E d_A^n\tr[\ketbra{\Lambda}_{ABA'B'}^{\otimes n}\,\rho_{A^n}^\intercal\otimes E_{B^n}] \chi(\ket{\Lambda}_{ABA'B'};\overline{Q_{\nu_E}^\delta})\,,\label{eq:Pfail_start}
\end{align}
where $\overline{Q_{\nu_E}^\delta}:=\tr_{A'B'}^{-1}(\overline{R_{\nu_E}^\delta})$ contains all the purifications of matrices in $\overline{R_{\nu_E}^\delta}$. In the following, we will bound~\eqref{eq:Pfail_start} independent of $\ket{\Lambda}\in\sP\sC$.

We first analyze the indicator function of the set $\overline{Q_{\nu_E}^\delta}$, which is by definition
\begin{align}
\chi(\ket{\Lambda}_{ABA'B'};\overline{Q_{\nu_E}^\delta}) = \begin{cases}
 1 & \text{if } \ket{\Lambda}_{ABA'B'} \in \overline{Q_{\nu_E}^\delta} \\
 0 & \text{otherwise}\,.
 \end{cases}
\end{align}
Without the knowledge of $\ket{\Lambda}_{ABA'B'}$, the condition $\ket{\Lambda}_{ABA'B'} \in \overline{Q_{\nu_E}^\delta}$ can only be \emph{physically checked} by a measurement POVM with effects $T$ and $\id-T$ acting on the quantum state $\ket{\Lambda}_{ABA'B'}$. Upon the observation of the effect $T$, we decide that $\ket{\Lambda}_{ABA'B'} \in \overline{Q_{\nu_E}^\delta}$ and similarly for $\id-T$. In other words, we are approximating $\chi(\ket{\Lambda}_{ABA'B'};\overline{Q_{\nu_E}^\delta})$ by a quantum measurement. Here we construct such an approximation using Holevo's covariant measurement~\cite{Holevo_book}.

Let $k$ be the number of copies of $\ket{\Lambda}\in\cH_{ABA'B'}$ used in the approximation, i.e.\@ we are given $\ketbra{\Lambda}^{\otimes k}$. If we ignore the fact that $\ket{\Lambda}\in\sP\sC$, we can use the Holevo's continuous POVM $\{s_{k,d^2_{AB}}\ketbra{\phi}^{\otimes k}\diff\phi\}$ to distinguish $\ket{\Lambda}\in\cH_{ABA'B'}$ among the set of pure states. Here, $\diff\phi$ is the uniform spherical measure on the set of pure states of $\cH_{ABA'B'}$ and $s_{k,d^2_{AB}}$ is the dimension of the symmetric subspace of $\cH_{ABA'B'}^{\otimes k}$. Coarse graining this measurement, we can distinguish $\ket{\Lambda}_{ABA'B'} \in \overline{Q_{\nu_E}^\delta}$ versus $\ket{\Lambda}_{ABA'B'} \in Q_{\nu_E}^\delta$ by the following POVM with two effects (analogous to Ref.~\cite{CR12})
\begin{align}
T_{Q_{\nu_E}^{\delta/2}} := s_{k,d^2_{AB}} \int_{\overline{Q_{\nu_E}^{\delta/2}}} \ketbra{\phi}^{\otimes k} \diff\phi \,, \textrm{ and } \id - T_{Q_{\nu_E}^{\delta/2}}\,.
\end{align}
We now check that this POVM indeed approximates $\chi(\ket{\Lambda};\overline{Q_{\nu_E}^\delta})$. For all $\ket{\Lambda}\in\overline{Q_{\nu_E}^\delta}$, using the definition of $\chi(\ket{\Lambda};\overline{Q_{\nu_E}^\delta})$
\begin{align}
\chi(\ket{\Lambda};\overline{Q_{\nu_E}^\delta}) - \tr(\ketbra{\Lambda}^{\otimes k}T_{Q_{\nu_E}^{\delta/2}}) &= 1-s_{k,d^2_{AB}} \int_{\overline{Q_{\nu_E}^{\delta/2}}} \tr(\ketbra{\Lambda}^{\otimes k}\ketbra{\phi}^{\otimes k}) \diff\phi\,.
\end{align}
Since $\ketbra{\Lambda}^{\otimes k}$ is supported on the symmetric subspace, we reinterpret the constant $1$ above as
\begin{align}
1= \tr\left(\ketbra{\Lambda}^{\otimes k}s_{k,d^2_{AB}} \int\ketbra{\phi}^{\otimes k} \diff\phi\right)\,,
\end{align}
which implies for all $\ket{\Lambda}\in\overline{Q_{\nu_E}^\delta}$
\begin{align}
\chi(\ket{\Lambda};\overline{Q_{\nu_E}^\delta}) - \tr(\ketbra{\Lambda}^{\otimes k}T_{Q_{\nu_E}^{\delta/2}}) &= s_{k,d^2_{AB}} \left(\int \tr(\ketbra{\Lambda}^{\otimes k}\ketbra{\phi}^{\otimes k}) \diff\phi- \int_{\overline{Q_{\nu_E}^{\delta/2}}} \tr(\ketbra{\Lambda}^{\otimes k}\ketbra{\phi}^{\otimes k}) \diff\phi \right) \\
&= s_{k,d^2_{AB}} \int_{Q_{\nu_E}^{\delta/2}} \tr(\ketbra{\Lambda}\ketbra{\phi})^k \diff\phi \\
&\leq s_{k,d^2_{AB}} \max_{\ket{\phi}\in Q_{\nu_E}^{\delta/2}} F(\Lambda_{AB},\tr_{A'B'}\ketbra{\phi}))^k\,.
\end{align}
By the definition of the sets
\begin{align}
R_{\nu_E}^{\delta/2} := \{\Psi\in\ChoiH:\exists \Psi'\in R_{\nu_E} \textrm{ with } P(\Psi,\Psi')\leq\delta/2\}\,,
\end{align}
and
\begin{align}
\overline{R_{\nu_E}^{\delta}} := \ChoiH \setminus \{\Psi\in\ChoiH:\exists \Psi'\in R_{\nu_E} \textrm{ with } P(\Psi,\Psi')\leq\delta\}\,,
\end{align}
we have for $\Lambda_{AB}\in\overline{R_{\nu_E}^\delta}$ and $\phi_{AB}:=\tr_{A'B'}\ketbra{\phi}\in R_{\nu_E}^{\delta/2}$
\begin{align}
F(\Lambda_{AB},\phi_{AB})=\sqrt{1-P(\Lambda_{AB},\phi_{AB})^2} \leq \sqrt{1-(\delta/2)^2} \leq e^{-\delta^2/2}\,,
\end{align}
using the reverse triangle inequality for purified distance. In summary, we obtain the approximation
\begin{align}
\chi(\ket{\Lambda};\overline{Q_{\nu_E}^\delta}) - \tr(\ketbra{\Lambda}^{\otimes k}T_{Q_{\nu_E}^{\delta/2}}) \leq \epsilon_1 := s_{k,d^2_{AB}}e^{-k\delta^2/2}\,. \label{eq:indicator_approx}
\end{align}

Now we can start bounding the failure probability. Inserting \eqref{eq:indicator_approx} into \eqref{eq:Pfail_start}, we have an intermediate bound
\begin{align}
P_{\textup{fail}}(\Lambda_{AB}) &\leq \epsilon_1 + \sum_E d_A^n\tr[\ketbra{\Lambda}_{ABA'B'}^{\otimes n}\,\rho_{A^n}^\intercal\otimes E_{B^n}] \tr[\ketbra{\Lambda}^{\otimes k}T_{Q_{\nu_E}^{\delta/2}}] \\
&= \epsilon_1 + \sum_E d_A^n\tr[\ketbra{\Lambda}_{ABA'B'}^{\otimes (n+k)}\,\rho_{A^n}^\intercal\otimes E_{B^n}\otimes T_{Q_{\nu_E}^{\delta/2}}] \ .\label{eq:Pfail_quantum}
\end{align}
Using the operator inequality in the Proposition~\ref{pro:operatorinequality}, namely
\begin{align}
\ketbra{\Lambda}_{ABA'B'}^{\otimes(n+k)} \leq s_{n+k,d^2_{AB}}^2 \int_{\sP\sC}\diff \nu(\ket{\Psi})  \, \ketbra{\Psi}^{\otimes (n+k)}\,,
\end{align}
we can bound the right hand side \emph{independent of the unknown} $\Lambda_{AB}$ as follows
\begin{align}
P_{\textup{fail}}(\Lambda_{AB}) &\leq \epsilon_1 + s_{n+k,d^2_{AB}}^2 \sum_E  \int_{\sP\sC}\diff \nu(\ket{\Psi}) d_A^n\tr[\ketbra{\Psi}^{\otimes n} \rho_{A^n}^\intercal\otimes E_{B^n}] \tr[\ketbra{\Psi}^{\otimes k}T_{Q_{\nu_E}^{\delta/2}}]\\
&= \epsilon_1 + s_{n+k,d^2_{AB}}^2 \sum_E c'_E \int \diff \nu_E(\Psi) \tr[\ketbra{\Psi}^{\otimes k}T_{Q_{\nu_E}^{\delta/2}}]\,,
\end{align}
where the last equality follows from the definition of the \emph{a posteriori} measure $\diff \nu_E(\Psi)$. For each measurement outcome $E$, the integral can split into two parts based on the set $R_{\nu_E}$ from which the kernels are uniformly bounded as follows:
\begin{align}
\int_{R_{\nu_E}} \diff \nu_E(\Psi) \tr(\ketbra{\Psi}^{\otimes k}T_{Q_{\nu_E}^{\delta/2}}) &\leq s_{k,d_{AB}^2}(1-(\delta/2)^2)^{k/2} \leq s_{k,d_{AB}^2}e^{-k\delta^2/2}\,,
\end{align}
using the definition of $T_{Q_{\nu_E}^{\delta/2}}$ and the fidelity bound $F(\Psi_{AB}\in R_{\nu_E},\phi_{AB}\in\overline{R_{\nu_E}^{\delta/2}})\leq\sqrt{1-(\delta/2)^2}$, and
\begin{align}
\int_{\overline{R_{\nu_E}}} \diff \nu_E(\Psi) \tr(\ketbra{\Psi}^{\otimes k}T_{Q_{\nu_E}^{\delta/2}}) &\leq \int_{\overline{R_{\nu_E}}}\diff \nu_E(\Psi)\,,
\end{align}
since $\tr(\ketbra{\Psi}^{\otimes k}T_{Q_{\nu_E}^{\delta/2}})\leq1$.
Choose $k=n$, the fact that $\sum_E c'_E\leq 1$, and combine all the inequalities together we have
\begin{align}
P_{\textup{fail}}(\Lambda_{AB}) &\leq \epsilon_1 + s_{n+k,d^2_{AB}}^2\epsilon_1 + s_{n+k,d^2_{AB}}^2 \sum_E c'_E \int_{\overline{R_{\nu_E}}}\diff \nu_E(\Psi) \\
&= s_{n,d^2_{AB}}e^{-n\delta^2/2} + s_{2n,d^2_{AB}}^2s_{n,d^2_{AB}}e^{-n\delta^2/2} + s_{2n,d^2_{AB}}^2 \sum_E c'_E \int_{\overline{R_{\nu_E}}}\diff \nu_E(\Psi) \\
&\leq s_{2n,d^2_{AB}}^3e^{-n\delta^2/2} + s_{2n,d^2_{AB}}^2 \sum_E c'_E \int_{\overline{R_{\nu_E}}}\diff \nu_E(\Psi)\,. \label{eq:Pfail_final}
\end{align}
If we choose $R_{\nu_E}$ and $\delta$ such that
\begin{align}
\int_{R_{\nu_E}}\diff \nu_E(\Psi) \geq 1-\frac{\epsilon}{2}s_{2n,d^2_{AB}}^{-2}\, \textrm{ and } \, \delta^2 = \frac{2}{n}\left(\ln\frac{2}{\epsilon}+3\ln s_{2n,d^2_{AB}}\right)
\end{align}
then $P_{\textup{fail}}(\Lambda_{AB})\leq \epsilon/2 + \epsilon/2 = \epsilon$ as desired. The proof of the Proposition is complete.
\end{proof}

\subsection{Regions for figures-of-merit}
The construction of confidence region on channel-space can be pushed-forward to obtain confidence regions for any figure-of-merit of channels we are interested in. The idea is exactly the same as reference~\cite{FR16} and we include it here for completeness. Let $f_\mathrm{channel}:\ChannelH\rightarrow\mathbb{R}$ be an arbitrary figure-of-merit of channels. The measure $\diff \nu_E(\Lambda)$ can be pushed-forward by $f_\mathrm{channel}$ to a measure on $\mathbb{R}$, which can then be represented as a density function $h(v)$ with respect to the Lebesgue measure of $\mathbb{R}$. Concretely, we have
\begin{align}
h(v)=\int\diff \nu_E(\Lambda)\delta(f_\mathrm{channel}(\Lambda)-v),
\label{eq:channel-space-h-of-v}
\end{align}
where $\delta(f_\mathrm{channel}(\Lambda)-v)$ is the Dirac delta measure on $\mathbb{R}$ at the point mass $v\in\mathbb{R}$. And for some subset of values $V$, the measure of $V$ is given by
\begin{align}
\int_{f_\mathrm{channel}^{-1}(V)} \diff \nu_E(\Lambda) = \int_V h(v)\diff v
\end{align}
where $\diff v$ is the Lebesgue measure on $\mathbb{R}$. The density $h(v)$ allows us to construct confidence interval for the property we desired.

\begin{proposition}
Let $f_\mathrm{channel}$ be a figure-of-merit and choose a confidence level $1-\epsilon$. For each dataset $E$, let $V_{\nu_E}\subseteq \mathbb{R}$ be a region of values such that
\begin{align}
    \int_{V_{\nu_E}} h(v)\diff v\geq 1-\frac{\epsilon}{2}s_{2n,d^2_{AB}}^{-2}\,,
\end{align}
and let $V_{\nu_E}^\delta$ be defined as
\begin{align}
V_{\nu_E}^\delta:=\{v\in \mathbb{R}:\exists v'\in V_{\nu_E} \textrm{ with } |v-v'|\leq\omega_{f_\mathrm{channel}}(\delta)\}\ ,
\end{align}
where $\omega_f(\delta):=\sup_{P(\Lambda,\Lambda')\leq\delta}|f(\Lambda)-f(\Lambda')|$. Then the mapping $E\mapsto V_{\nu_E}^\delta$ is a confidence region estimator for the figure-of-merit $f_\mathrm{channel}$ with confidence level $1-\epsilon$ if
\begin{align}
\delta^2 = \frac{2}{n}\left(\ln\frac{2}{\epsilon}+3\ln s_{2n,d^2_{AB}}\right)
\end{align}
In other words, for all channel $\Lambda\in\ChannelH$
\begin{align}
       \Pr_E[f_\mathrm{channel}(\Lambda)\in V_{\nu_E}^\delta] \geq 1-\epsilon.
\end{align}
\end{proposition}

\begin{proof}
It is clear from the fact that as defined, $V_{\nu_E}^\delta\supseteq f_\mathrm{channel}(f_\mathrm{channel}^{-1}(V_{\nu_E})^\delta)$.
\end{proof}

For each figure-of-merit of interest, we can derive a bound on $\omega_f(\delta)$ by simple inequalities for distance measures. For diamond distance, we have the following result.
\begin{proposition}
\label{thm:channel-space-diamond}
For each dataset $E$, let $\gamma_E\in[0,1]$ be such that
\begin{align}
    \int_0^{\gamma_E} h(v)\diff v\geq 1-\frac{\epsilon}{2}s_{2n,d^2_{AB}}^{-2}\,,
\end{align}
Then the mapping $E\mapsto [0,\gamma_E+d_1\delta/2]$ is a confidence region estimator for the diamond distance to ideal with confidence level $1-\epsilon$ if
\begin{align}
\delta^2 = \frac{2}{n}\left(\ln\frac{2}{\epsilon}+3\ln s_{2n,d^2_{AB}}\right)
\end{align}
In other words, for all channel $\Lambda\in\ChannelH$
\begin{align}
       \Pr_E\left[\frac{1}{2}||\Lambda_{A\to B}-\Lambda^{\mathrm{ideal}}_{A\to B}||_\diamond \leq \gamma_E + d_A\delta/2 \right] \geq 1-\epsilon\,,
\end{align}
where the probability is over the random dataset $E$ with distribution $\Pr(E|\Lambda)=\mathcal{L}(\Lambda|E)$.
\end{proposition}
\begin{proof}
Continuing from the previous Proposition, we set $V_E:=[0,\gamma_E]$; it remains for us to obtain a bound on $\omega_{f_\diamond}(\delta)$. Using the reverse triangle inequality and SDP reformulation of diamond norm, we have
\begin{align}
|f_\diamond(\Lambda)-f_\diamond(\Lambda')| &= \frac{1}{2}\left| ||\Lambda-\Lambda^{\mathrm{ideal}}||_\diamond - ||\Lambda'-\Lambda^{\mathrm{ideal}}||_\diamond \right| \nonumber \\
&\leq \frac{1}{2}||\Lambda_{A\to B}-\Lambda'_{A\to B}||_\diamond \\
&\leq \frac{1}{2}||d_A(\Lambda_{AB}-\Lambda'_{AB})||_1\,,
\end{align}
where the last inequality utilises the duality between Schatten $1$-norm and Schatten $\infty$-norm to bound the objective function of the diamond norm SDP. Since the purified distance dominates the trace distance, we obtain
\begin{align}
\frac{1}{2}||\Lambda_{AB}-\Lambda'_{AB}||_1 \leq \frac{1}{2}P(\Lambda_{AB},\Lambda'_{AB})\,,
\end{align}
which implies $\omega_{f_\diamond}(\delta)\leq d_A\delta/2$.
\end{proof}

For worst-case entanglement fidelity, we have the following result.
\begin{proposition}
\label{thm:channel-space-diamond}
For each dataset $E$, let $\gamma_E\in[0,1]$ be such that
\begin{align}
    \int_0^{\gamma_E} h(v)\diff v\geq 1-\frac{\epsilon}{2}s_{2n,d^2_{AB}}^{-2}\,,
\end{align}
Then the mapping $E\mapsto [0,\gamma_E-d_A\delta]$ is a confidence region estimator for the diamond distance to ideal with confidence level $1-\epsilon$ if
\begin{align}
\delta^2 = \frac{2}{n}\left(\ln\frac{2}{\epsilon}+3\ln s_{2n,d^2_{AB}}\right)
\end{align}
In other words, for all channel $\Lambda\in\ChannelH$
\begin{align}
       \Pr_E\left[F_{\mathrm{worst}}(\Lambda_{A\to B}) \geq \gamma_E - d_A\delta \right] \geq 1-\epsilon\,,
\end{align}
where the probability is over the random dataset $E$ with distribution $\Pr(E|\Lambda)=\mathcal{L}(\Lambda|E)$.
\end{proposition}
\begin{proof}
We set $V_E:=[\gamma_E,1]$. Let $\rho_A$ be an optimizer of $F_{\mathrm{worst}}(\Lambda')$, since $\rho_A$ will give an upper bound on $F_{\mathrm{worst}}(\Lambda)$ we have with $f=F_{\mathrm{worst}}$
\begin{align}
|f(\Lambda)-f(\Lambda')| &= |F_{\mathrm{worst}}(\Lambda)-F_{\mathrm{worst}}(\Lambda')| \leq |\bra{\tilde{\Phi}}\rho_A(d_A\Lambda_{AB})\rho_A\ket{\tilde{\Phi}} - \bra{\tilde{\Phi}}\rho_A(d_A\Lambda'_{AB})\rho_A\ket{\tilde{\Phi}}| \\
&=d_A\left|\ip{\rho_A\ketbra{\tilde{\Phi}}\rho_A}{\Lambda_{AB}-\Lambda'_{AB}}\right|  \leq d_A \norm{\rho_A\ketbra{\tilde{\Phi}}\rho_A}_{\infty} \norm{\Lambda_{AB}-\Lambda'_{AB}}_1 \leq d_A\delta\,,
\end{align}
using Holder inequality for Schatten norms and $\norm{\Lambda_{AB}-\Lambda'_{AB}}_1 \leq P(\Lambda_{AB},\Lambda'_{AB})$.
\end{proof}

\section{\label{sec:numerics}Metropolis-Hastings algorithm in channel space}
The previous two sections describe the construction of confidence region estimators for quantum processes, which utilize distributions $\diff\mu_E(\sigma)$ and $\diff\nu_E(\Lambda)$. We now describe how one can numerically estimate such distributions so that the densities $\mu(v)$ and $h(v)$ can be approximated.

The distribution $\diff\mu_E(\sigma)$ or the density $\mu(v)$ can be estimated by numerically producing a lot of samples. These can be generated by the Metropolis-Hastings random walk in (bipartite) state space, whose details can be found in Ref.~\cite{FR16}. Here we only discuss the Metropolis-Hastings random walk in channel space.

Recall that in the channel space method, we need to be able to compute the density $h(v)$ for the given figure-of-merit $f_\mathrm{channel}$. We do this numerically using Metropolis-Hastings algorithm. The output of this algorithm is a histogram of the figure-of-merit which approximates the continuous density.

Let us recall the Metropolis-Hasting algorithm for continuous sample space~\cite{MHalgorithm}. Let $p(x)\diff x$ be the target distribution from which we want to sample, and $q(x'|x)\diff x'$ be a proposal distribution, all displayed with respect to the same base measure $\diff x = \diff x'$. We assume that the proposal density function is symmetric $q(x'|x)=q(x|x')$. When the process is at point $x$, the distribution $q(x'|x)\diff x'$ proposes a new point $x'$. If $p(x')/p(x)\geq1$ then we jump unconditionally to the new point $x'$; otherwise, $p(x')/p(x)<1$ and we jump to $x'$ only with probability $p(x')/p(x)$. The points visited in this fashion, for a large number of iterations, are distributed according to the target distribution. Note that the algorithm only requires computing the ratio $p(x')/p(x)$ and thus does not require determining any normalization factor for $p(x)$.

We want to generate samples from the target distribution
\begin{align}
\diff\nu_E(\Lambda):=c'^{-1}_E\mathcal{L}(\Lambda|E)\diff \nu(\Lambda)
\end{align}
where $\mathcal{L}(\Lambda|E)$ is the prepare-and-measure or ancilla-assisted likelihood function and $\diff \nu(\Lambda)$ is the induced measure on channel space. Recalling the definition of $\diff \nu(\Lambda)$, we thus want to sample from
\begin{align}
\diff\nu_E(U_{BA'B'}) = c'^{-1}_E\mathcal{L}(U_{BA'B'}|E)\diff U_{BA'B'}\,,
\end{align}
with $\diff U_{BA'B'}$ the invariant Haar measure. Concretely, in the prepare-and-measure scheme we take
\begin{align}
\mathcal{L}_{\text{PM}}(U|E) = d_A^n\tr\bigg( (U\ketbra{\Psi_0}U^\dagger)^{\otimes n} \, \bigotimes_{j,k,\ell} (\sigma_A^j)^\intercal \otimes E^{(\ell)}_k \bigg) 
\end{align}
and in ancilla-assisted scheme we take
\begin{align}
\mathcal{L}_{\text{AA}}(U|E) = d_A^n\tr\bigg((U\ketbra{\Psi_0}U^\dagger)^{\otimes n} \,\bigotimes_{k,\ell} \psi_P^{1/2}\,E_k^{(\ell)}\,\psi_P^{1/2} \bigg)\,,
\end{align}
where $\ket{\Psi_0}$ is the fixed reference state in~\eqref{eq:fixedpurifiedChoi}. This can be done using the Metropolis-Hastings algorithm, by designing a symmetric proposal distribution over the space of all unitaries $U_{BA'B'}$ and setting $q(U'_{BA'B'}|U_{BA'B'})\propto \mathcal{L}(U'_{BA'B'}|E)$. To ensure $q(U'|U)=q(U|U')$, let $q(W)\diff W$ be a distribution on unitaries on $BA'B'$ such that $q(W)=q(W^\dagger)$. For each point $U$, if we define $U':=WU$, then we have a symmetric proposal distribution $q(U'|U)=q(WU|U)=q(W)=q(W^{-1})=q(W^{-1}U'|U')=q(U|U')$, namely $q(WU|U)\diff W$ where $\diff W$ is the Haar measure. It remains to fix a $q(W)\diff W$ with $q(W)=q(W^\dagger)$. We have implemented two choices:
\begin{itemize}
    \item ``$e^{iH}$-type jumps'': We pick a random $d_{BA'B'}\times d_{BA'B'}$ matrix $N$ with each entry independent and normally distributed complex numbers with standard deviation given by the step size. We then calculate $H=N+N^\dagger$ and set $W=e^{iH}$, inducing a measure $q(W)\,dW$. Denoting by $dN$ the measure induced on $N$ by this sampling procedure, observe that $dN=d(-N)$ as the normal distribution is symmetric.  Furthermore the Haar measure is invariant under the adjoint, $dW = d(W^\dagger)$, since $d(W^\dagger)$ is also unitarily invariant and is thus also the Haar measure.  Hence, $q(W)\,dW = dN = d(-N) = q(W^\dagger)\,d(W^\dagger) = q(W^\dagger)\,dW$ as required.
    \item ``elementary rotation jumps'': Choose $m\in\{x,y,z\}$ uniformly at random and choose two indices $i<j$ uniformly at random.  Choose $\sin(\alpha)$ at random (normally distributed number whose standard deviation is the step size; truncated to $[-1, 1]$). Define the unitary $W_1$ as the qubit rotation on the subspace spanned by $\{\ket i,\ket j\}$ defined by $e^{i\alpha\,(\vec{e}_m\cdot\vec{\sigma})} = \cos(\alpha)\, \Ident + i\sin(\alpha)\,(\vec{e}_m\cdot\vec{\sigma})$, where $\vec{e}_m$ is the $m$-th basis vector in 3D and where $\{\sigma_x,\sigma_y,\sigma_z\}$ are the Pauli matrices. We see that $-\alpha\,(\vec e_m\cdot\vec\sigma)$ is sampled with the same probability as $\alpha\,(\vec e_m\cdot\vec\sigma)$ and hence for the same reason as above, $q(W)=q(W^\dagger)$.  In order to keep the acceptance ratio at a reasonable rate, we sample $N_\text{inner-iter}$ different instances of $W_1$, and multiply them together to form the sampled $W$.  One should choose $N_\text{inner-iter}$ such that it is possible to keep the acceptance rate around $30\%$.  
\end{itemize}

\section{\label{sec:convergence}Convergence in number of samples $N\to\infty$}

We now turn to an example where we clearly observe the convergence of the distributions $h(v)$ and $\mu(v)$ around the known true figure-of-merit.  Consider a noisy identity process on a qutrit, of the form
\begin{align}
  \Lambda_{A\to B}(\rho) = p\,\rho + (1-p)\,d_B^{-1}\id_B\ ,
\end{align}
with $p=0.96$ and $d_B=3$. This gives us a diamond norm to the identity process of
\begin{align}
  \frac12\norm{\Lambda_{A\to B} - \mathcal{I}_{A\to B}}_\diamond
  = 0.03556\ .
\end{align}

We consider measurements on the input and output systems given by using the
Gell-Mann matrices as observables:
\begin{align*}
\lambda_{1} &= {\begin{pmatrix}0&1&0\\1&0&0\\0&0&0\end{pmatrix}}\,; &
\lambda_{2} &= {\begin{pmatrix}0&-i&0\\i&0&0\\0&0&0\end{pmatrix}}\,; &
\lambda_{3} &= {\begin{pmatrix}1&0&0\\0&-1&0\\0&0&0\end{pmatrix}}\,; \\
\lambda_{4} &= {\begin{pmatrix}0&0&1\\0&0&0\\1&0&0\end{pmatrix}}\,; &
\lambda_{5} &= {\begin{pmatrix}0&0&-i\\0&0&0\\i&0&0\end{pmatrix}}\,; &
\lambda_{6} &= {\begin{pmatrix}0&0&0\\0&0&1\\0&1&0\end{pmatrix}}\,; \\
\lambda_{7} &= {\begin{pmatrix}0&0&0\\0&0&-i\\0&i&0\end{pmatrix}}\,; &
\lambda_{8}
&= \frac{1}{\sqrt{3}}{\begin{pmatrix}1&0&0\\0&1&0\\0&0&-2\end{pmatrix}}\,.
\hspace*{-2em}
\end{align*}
Each single-system measurement setting has three possible outcomes.  For each
pair of measurement settings (for the input and the output system) we simulate
$N$ measurement outcomes.  We choose $N=10^6$ for our reference experiment,
yielding a total of $N_\mathrm{tot} = 8^2\times 10^6 = 6.4\times 10^7$
measurement outcomes.  We denote the corresponding frequency vector by
$(n^{\mathrm{Ref}}_{j_Aj_B,\ell_A\ell_B})$, where $j_i$ labels the measurement
setting on system $i$ and $\ell_i$ labels the corresponding measurement
outcome. We group together all the indices into a collective index $k$, such
that $n^{\mathrm{Ref}}_k$ denotes the number of times the joint POVM effect
$E^{(k)}_{AB}$ was observed.

The corresponding analysis is depicted in \autoref{fig:ConvergenceNInfty}, as
the curve labeled ``100\%''.  Thanks to the large number of measurements, the
distributions $h(v)$ and $\mu(v)$ peak sharply around the true value of
$f_\diamond$.

We now ask, how would these distributions look if fewer measurements had
been taken?  Instead of simulating new outcomes, which would cause the peak to
be displaced and would make a comparison more difficult, we artificially rescale
the frequency vector $n^{\mathrm{Ref}}_k$ by a factor $\alpha$, i.e., we define
$n^{\alpha}_k = \lfloor\alpha\,n^{\mathrm{Ref}}_k\rfloor$, where by
$\lfloor x\rfloor$ we denote the largest integer less than or equal to $x$.  For
instance, we may choose $\alpha=0.01=1\%$ to represent an experiment in which
only $n \approx \alpha\, N = 10^4$ measurements per setting were sampled,
instead of $N$.  While this rescaling of the frequency vector is artificial, the
resulting measurement counts are still representative of possible outcomes that
one could have sampled if we had simulated directly only $\alpha\,N$ outcomes
per setting; crucially, doing so facilitates comparisons between the different
settings.  The analysis for a selection of values for $\alpha$ (given as
percentages) is presented in \autoref{fig:ConvergenceNInfty}.  The corresponding
peaks are indeed seen to converge towards the true value of $f_\diamond$.
For each value of $\alpha$, we calculate the corresponding quantum error bars
$(v_0, \Delta, \gamma)$, and plot them against $\alpha$
(\autoref{fig:ConvergenceNInfty}, inset).  The quantum error bars become a
  better and tighter description of the true state as the number of measurements
  increase, as expected.

The quantum error bar $\Delta$ is the one which is most akin to a
  ``standard deviation,'' as in the limit $\gamma\to0$ the fit
  model~\eqref{eq:fit-model-1} becomes a Gaussian.  We may investigate the
  precise scaling of $\Delta$ as a function of the number of measurements by
  plotting the magnitude of this quantum error bar against the number of
  measurements in a log-log plot (\autoref{fig:convergence_plot_scaling}).  We
  indeed observe a scaling close to $1/\sqrt{n}$, where $n\approx\alpha N$ is
  the number of measurements, as expected from known results in usual quantum
  tomography.  We expect that by improving the measurement settings, for
  instance by using adaptive measurements, tighter error bars can be achieved
  with fewer
  measurements~\cite{Sugiyama2012PRA_AdaptiveTomo,Mahler2013PRL_adaptive,Granade2016arXiv_adaptive,Pogorelov2017PRA_adaptive}.

  This depiction allows us again to
appreciate the convergence to the true value of $f_\diamond$.

\begin{figure}
  \centering
  \includegraphics{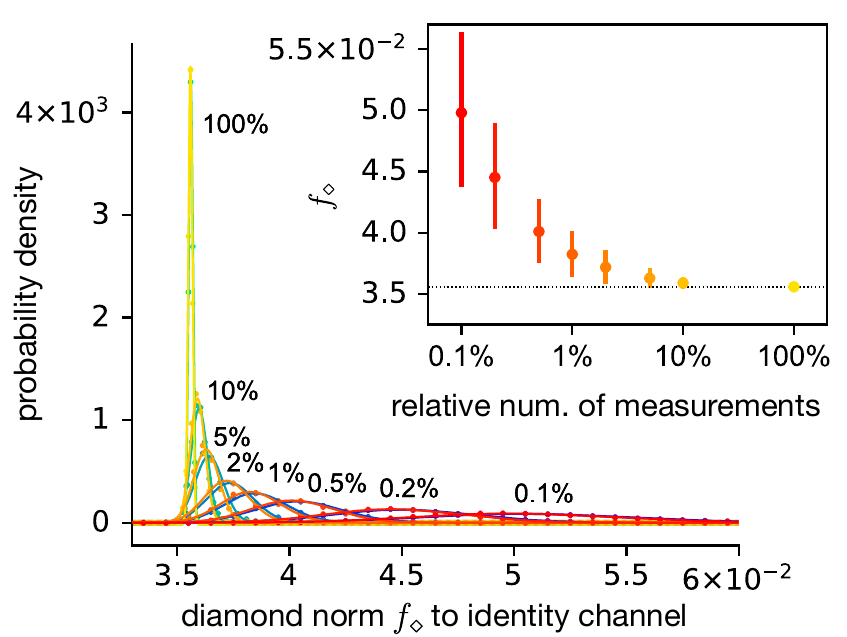}
  \caption{Convergence of quantum error bars to the true value of figure of
    merit in the limit of many measurements, for a noisy identity process on a
    qutrit.  Measurements using Gell-Mann matrices as observables on the input
    and the output systems were simulated with $10^6$ outcomes per setting,
    providing the reference experiment (labeled ``100\%'').  Analyses as in
    \autoref{fig:results-example-plots} were then carried out after artificially
    rescaling the measured frequency counts by various factors (percentage
    labels), allowing us to compare regimes with different number of
    measurements while still keeping the estimated expectation values of the
    measured observables constant to facilitate comparison.  As the number of
    measurements increases, the distribution of $f_\diamond$, the diamond norm
    to the identity channel, peaks to the known true value of
    $3.556\times 10^{-2}$.  Data points display the numerical histogram
    (biparite-state sampling method: blue--green; channel-space method:
    red--yellow) which are fit to our model \#1.  \textbf{Inset:} the quantum
    error bars ($v_0,\Delta,\gamma$) obtained from the fit~\protect\cite{FR16}
    (channel-space method only) are plotted against the number of
    measurements relative to the reference experiment; markers represent $v_0$
    with an error bar representing $[v_0-(\Delta-\gamma),v_0+\Delta+\gamma]$ for
    each analysis instance.  The dotted line indicates the known true value of
    $f_\diamond$.}
  \label{fig:ConvergenceNInfty}
\end{figure}

\begin{figure}
  \centering
  \includegraphics[width=86mm]{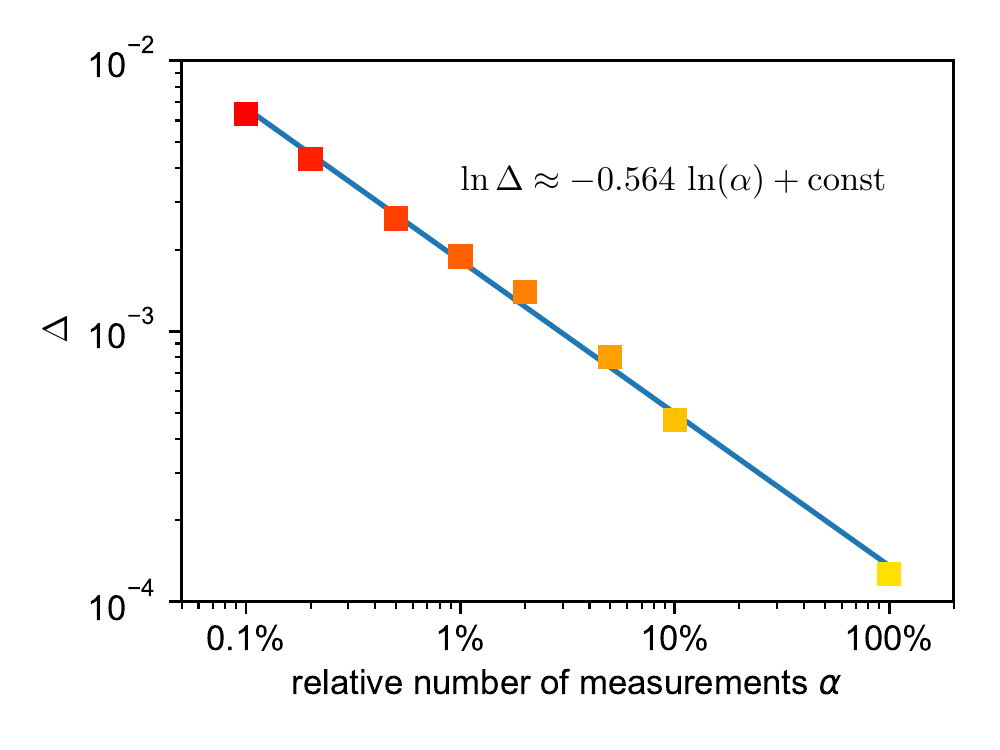}
  \caption{One of the quantum error bars, $\Delta$, is observed to scale
    approximately as $1/\sqrt{N}$, where $N$ indicates the number of
    measurements, as expected in standard (non-adaptive) quantum tomography.
    The setting is the same as in \autoref{fig:ConvergenceNInfty}.  By choosing
    more sophisticated measurement operators, for instance by adapting the
    measurement settings based on earlier outcomes, the scaling could be
    improved~\cite{Sugiyama2012PRA_AdaptiveTomo,Mahler2013PRL_adaptive,Granade2016arXiv_adaptive,Pogorelov2017PRA_adaptive}.}
  \label{fig:convergence_plot_scaling}
\end{figure}

\section{\label{sec:compare}Relations of two methods}
In this section, we discuss the theoretical connections between the two methods, specifically the relationship between the densities $\mu(v)$ and $h(v)$. We will use basic notions from measure theory which is available in any standard textbook.

Recall that we use the induced measure $\diff\sigma_{AB}$ on density matrices $\rmD(\cH_{AB})$ in the biparite-state sampling and the measure $\diff \nu(\Psi)$ on Choi state $\sC(\cH_{AB})$ in the channel space method. It is helpful for the reader to refresh the definition of these measures in Appendix~\ref{sec:measures}. The following result connect these probability measures; its proof is delayed till the end of this Appendix.

\begin{proposition}
\label{prop:relation_measures}
The measure $\diff \sigma_{AB}$ factors as $\diff\sigma_A \diff \nu(\Lambda_{A\to B})$ in the sense that for all measurable function $g(\sigma_{AB})$
\begin{align}
\int \diff\sigma_{AB}\, g(\sigma_{AB}) &= \int \diff \sigma_A\int\diff \nu(\Lambda_{A\to B})\, g(d_A \sigma_A^{1/2}J(\Lambda_{A\to B})\sigma_A^{1/2})\\
 &= \int \diff \sigma_A\int\diff \nu(\Lambda_{A\to B})\, g(d_A \sigma_P^{1/2}J(\Lambda_{A\to B})\sigma_P^{1/2})\label{eq:relation_measures}\,,
\end{align}
where $\diff\sigma_A$ is the reduced measure of $\diff\sigma_{AB}$ via partial tracing and $\diff\nu(\Lambda_{A\to B})$ is the uniform measure on channel space induced by $\diff U_{BA'B'}$ and $\sigma_P=\sigma_A^\intercal$.
\end{proposition}
We remark that intuitively this result is clear: the probability measure $\diff\sigma_{AB}$ can be ``conditioned'' on different values of $y=\tr_B(\sigma_{AB})$ giving rise to conditional probability measures $\diff\nu_y(\sigma_{AB})$ and these are recognised as $\diff \nu(\Lambda)$ by unitary invariance. However, the fact that these events which we are conditioning on has measure zero under $\diff\sigma_{AB}$ makes the proof more complicated.

Proposition~\ref{prop:relation_measures} tells us that integrating over all bipartite states according to the measure $\diff\sigma_{AB}$ can be done by separately integrating over all possible input states $\sigma_A$ and over all possible channels $\Lambda_{A\to B}$, by combining them as $\sigma_A^{1/2}\,\Lambda_{AB}\,\sigma_A^{1/2}$ where $\Lambda_{AB} = J(\Lambda_{A\to B})$. Equivalently, this can be done by separately integrating over all possible \emph{transposed} input states $\sigma_A$ and over all possible channels $\Lambda_{A\to B}$ as in~\eqref{eq:relation_measures}.  We can use this intuition to relate the two methods presented above.

\begin{figure} 
  \includegraphics[width=0.4\columnwidth]{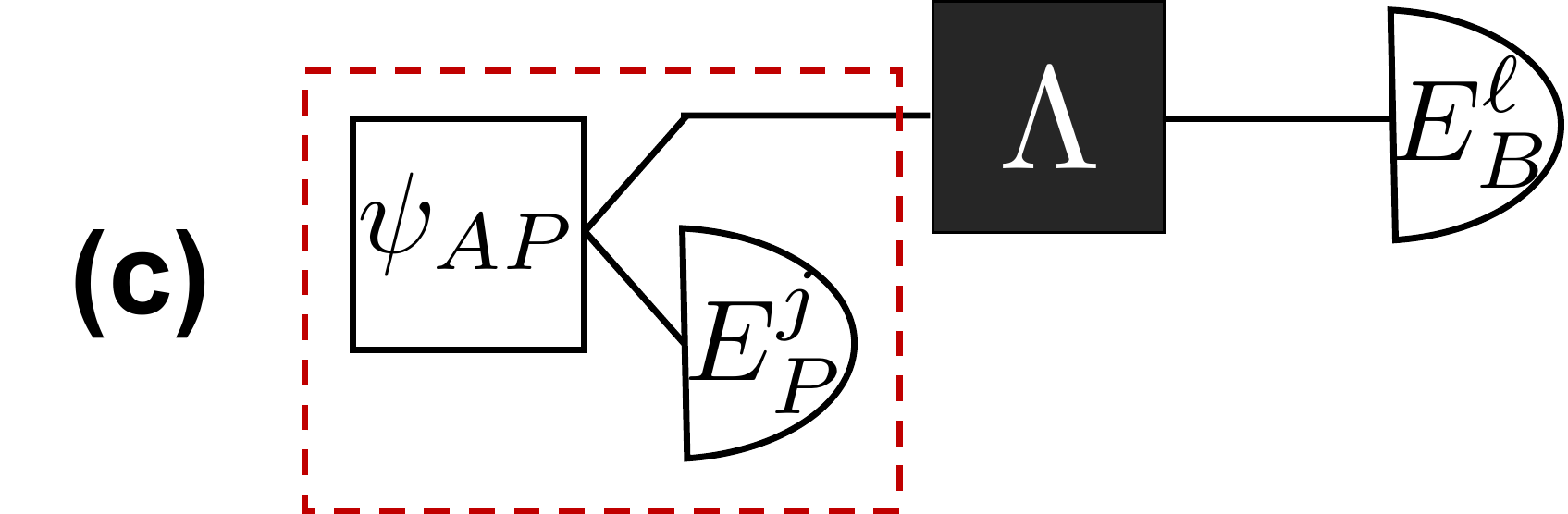}
  \caption{\label{fig:setups_relationship}An intermediate tomographic scheme. Scenario (c) comes from restricting $E^{\ell}_k$ acting on $BP$ of \autoref{fig:tomo}\textbf{(b)} to be a tensor product measurement. The measurement $E_P^j$ on half of an entangled state in (c) can be seen as a probabilistic state preparation similar to \autoref{fig:tomo}\textbf{(a)}.}
\end{figure}

In order to connect both quantities, we consider the situation depicted in
  \autoref{fig:setups_relationship}.  Assume that for each
  repetition $j=1\dots n$ the input $\rho_A^j$ is chosen by a measurement on the pure state
  $\ket\psi_{AP} = \sigma_A^{1/2}\,\ket{\hat\Phi}$ for some given state
  $\sigma_A$, and the outcome POVM effect $E_P^j$ was observed. The measurement on the output state of the channel is chosen from some collection of measurements acting only on system $B$ only. Assuming that the outcome POVM effect $E^j_B$ was observed, the dataset $E$ consists of the pairs $(E^j_P,E^j_B)$ for all $n$ repetitions.

 Viewing this scenario as an ancilla-assisted scheme (by moving the measurement on $P$ to the end), we can employ the biparite-state sampling method and
  calculate $\mu(v)$ by integrating our test function $\delta(f(\rho_{AB})-v)$
  over the \emph{full biparite-state space} according
  to~\eqref{eq:method-naive-mu-of-f} and~\eqref{eq:def-mu-E-sigmaAB}:
  \begin{align}
    \mu(v) &= c_E^{-1}\int \diff\sigma_{AB}\,\mathcal{L}_1(\sigma_{AB}|E)\,\delta(f(\rho_{AB})-v)\ ,
    \label{eq:comparison-mu-v}
  \end{align}
  where
  \begin{align}
    \mathcal{L}_1(\sigma_{AB}|E)= \tr(\sigma_{AB}^{\otimes n}\,E) \,,
    \label{eq:comparison-mu-E-sigmaAB}
  \end{align}
for $E = \otimes_{j=1}^n E^j_P\otimes E^j_B$.

  On the other hand, we can also view this as a prepare and measure scheme and use the channel space method to compute, the histogram by~\eqref{eq:def-nu-E-channels}
  and~\eqref{eq:channel-space-h-of-v} as an integration over the \emph{space of all
  channels} only,
  \begin{align}
    h(v) = c'^{-1}_E\int \diff\nu(\Lambda_{A\to B})\mathcal{L}_2(\Lambda|E)\delta(f_\mathrm{channel}(\Lambda)-v)\ ,
    \label{eq:comparison-h-v}
  \end{align}
  where
  \begin{align}
  \mathcal{L}_2(\Lambda|E) = \prod_{j=1}^n \tr(\Lambda_{A\to B}(\rho_A^j)\,E_B^j)\,.
  \end{align}
  We may rewrite each factor term using the Choi-Jamiolkowski state of the channel as
  \begin{align}
  \tr(\Lambda_{A\to{}B}(\rho_A^j)\,E_B^j) =  \tr(\sigma_P^{1/2}\,\Lambda_{PB}\,\sigma_P^{1/2}
  \,(E_B^j\otimes E_P^j))
  \end{align}
   (where $\sigma_P = \sigma_A^T$) and thus
  \begin{align}
    \mathcal{L}_2(\Lambda|E;\sigma_A) = \tr( (\sigma_P^{1/2}\,\Lambda_{PB}\,\sigma_P^{1/2})^{\otimes n}\,E)\ ,
    \label{eq:comparison-nu-E-sigmaP-Lambda_PB}
  \end{align}
  now defining the same operator $E= \otimes_{j=1}^n E^j_P\otimes E^j_B$ as before and where $\sigma_A$ is fixed.

  The similarity of~\eqref{eq:comparison-mu-v}
  and~\eqref{eq:comparison-mu-E-sigmaAB} with~\eqref{eq:comparison-h-v}
  and~\eqref{eq:comparison-nu-E-sigmaP-Lambda_PB} is now more evident.  It is
  worth giving a precise interpretation to both $\mathcal{L}_1(\sigma_{AB}|E)$ and
  $\mathcal{L}_2(\Lambda|E;\sigma_A)$.  The function $\mathcal{L}_1(\sigma_{AB}|E)$ is a
  probability density on the biparite-state space with respect to
  $\diff\sigma_{AB}$, describing the Bayesian posterior distribution after observing
  data $E$ for an agent using the uniform prior $\diff\sigma_{AB}$ (and thus
  ignoring any prior information about what the input state actually is).  On
  the other hand, $\mathcal{L}_2(\Lambda|E;\sigma_A)$ is the posterior
  distribution in the space of all channels, after observing data $E$ for an
  agent which is using the prior $\diff\nu(\Lambda_{A\to B})$.  Yet,
  \autoref{prop:relation_measures} tells us that the prior $\diff\nu(\Lambda_{A\to B})$
  is precisely the same as the prior in the biparite-state space corresponding
  to knowing with certainty that the input state is exactly $\sigma_A$.  Indeed,
  $\diff\nu(\Lambda_{A\to B})$ is precisely the measure induced by
  $d\sigma'_{AB}\delta(\tr_B(\sigma'_{AB})-\sigma_A)$ on
  $\Lambda_{A\to B} = J^{-1}(\sigma_A'^{-1/2}\sigma'_{AB}\sigma_A'^{-1/2})$, where $\delta(\tr_B(\sigma'_{AB})-\sigma_A)$ is a Dirac delta at the point $\sigma_{A}$. That is, with the shorthand $\sigma'_A=\tr_B(\sigma'_{AB})$, we may rewrite~\eqref{eq:comparison-h-v} as
  \begin{align}
    h(v) = c'^{-1}_E\int d\sigma'_{AB} \delta(\sigma'_{A} - \sigma_A)\int\diff\nu(\Lambda_{A\to B}) \cdot \mathcal{L}_2(\Lambda|E;\sigma'_A)\delta(f_\mathrm{channel}(\Lambda_{A\to B}) - v)\ .
    \label{eq:comparison-h-v-rewritten}
  \end{align}

  Hence, the difference between the bipartite sampling method and the
  channel-space method, at least in the current scenario, is exactly the prior information about the input state.
  In the former, nothing is assumed about the input state other than what can be
  inferred directly from the measurement data; in the latter, the exact input
  state is assumed with certainty as represented by the first Dirac delta
  function in~\eqref{eq:comparison-h-v-rewritten}.\\
  
Finally, we will prove the following result, which is easily seen to imply the Proposition~\ref{prop:relation_measures}.
\begin{proposition}
\label{prop:relation_measures_detailed}
There exists an essentially unique family of probability measures $\diff\nu_y(\sigma_{AB})$ on $\rmD(\cH_{AB})$ indexed by full rank $y\in\rmD(\cH_A)$ such that
\begin{align}
\int \diff\sigma_{AB} g(\sigma_{AB}) = \int \diff\sigma_A(y)\int_{\tr_B^{-1}(y)}\diff\nu_y(\sigma_{AB})g(\sigma_{AB})\,,
\end{align}
where $\diff\sigma_A(y)$ is the reduced measure of $\diff\sigma_{AB}$ and $\tr_B^{-1}(y)$ denotes the preimage of $y$ under partial tracing $B$. Moreover, each member $\diff\nu_y(\sigma_{AB})$ of the family is supported on $\tr_B^{-1}(y)$ and actually isomorphic to $\diff \nu(\Psi)$ on $\ChannelH$. These isomorphisms are given by
\begin{align}
J_y^{-1}:\tr_B^{-1}(y)&\rightarrow \ChannelH \\
\sigma_{AB} &\mapsto J^{-1}\left(d_A^{-1}y^{-1/2}\sigma_{AB}y^{-1/2}\right)\,.
\end{align}
\end{proposition}
\begin{proof}
Again, it will be convenient to work in the purified picture. By definition, $\diff\sigma_{AB}$ originates from the uniform spherical measure $\diff\ket{\phi}_{ABA'B'}$ induced by the Haar measure $\diff U_{ABA'B'}$ by the relation $\ket{\phi}_{ABA'B'}=U_{ABA'B'}\ket{\Psi_0}$. On the other hand, $\diff \nu(\Psi)$ comes from the Haar measure $\diff U_{BA'B'}$ via the relation $\ket{\Psi}=U_{BA'B'}\ket{\Psi_0}$.

Consider the partial trace $\tr_{BA'B'}:\End(\cH_{ABA'B'})\rightarrow\End(\cH_A)$. Two things happen under this mapping.

First, the measure $\diff\ket{\phi}_{ABA'B'}$ admits a pushforward along $\tr_{BA'B'}$ denoted as $\diff\sigma_A(y)$ living on space $\rmD(\cH_A)$. Note that this measure $\diff\sigma_A(y)$ no longer coincides with the Haar induced (or Hilbert-Schmidt induced) uniform measure on $\rmD(\cH_A)$ (since such measure arises uniquely from the Haar measure $\diff U_{AA'}$ acting on $\cH_{AA'}$).

Second, the space $\End(\cH_{ABA'B'})$ is partitioned into fibers $\tr_{BA'B'}^{-1}(y)$ over $y\in\End(\cH_A)$. Observe that one of such fibers corresponds to the set of purified Choi states $\sP\sC$: take $y=\id/d_A$. Moreover, if $y\in\rmD(\cH_A)$ is full rank, then the fiber over $y$ is isomorphic to $\sP\sC$. Indeed, the bijection is given by
\begin{align}
J_y^{-1}:\tr_{BA'B'}^{-1}(y)&\rightarrow \sP\sC \\
\phi_{ABA'B'} &\mapsto d_A^{-1}y^{-1/2}\phi_{ABA'B'}y^{-1/2}\,.
\end{align}
Note that partial tracing out $A'B'$ gives Choi-Jamiolkowski isomorphisms identifying $\tr_B^{-1}(y)\subseteq\rmD(\cH_{AB})$ with the space of all quantum processs: 
\begin{align}
J_y^{-1}:\tr_B^{-1}(y)&\rightarrow \ChannelH \\
\sigma_{AB} &\mapsto J^{-1}\left(d_A^{-1}y^{-1/2}\sigma_{AB}y^{-1/2}\right)\,,
\end{align}
where $J^{-1}$ is the standard Choi-Jamiolkowski isomophism identifying $\ChoiH$ with $\ChannelH$. We stress again that these are isomorphisms only for full rank $y\in\rmD(\cH_A)$.

The probability measure $\diff\phi_{ABA'B'}$ then disintegrates~\cite{Disintegration} into a family of conditional probability measures $\diff\nu_y(\phi_{ABA'B'})$ on each fiber (or preimage over $y$) $\tr_{BA'B'}^{-1}(y)$ such that
\begin{align}
\int \diff \phi_{ABA'B'} g(\phi_{ABA'B'}) = \int \diff\sigma_A(y)\int_{\tr_{BA'B'}^{-1}(y)}\diff\nu_y(\phi_{ABA'B'})g(\phi_{ABA'B'})
\end{align}
for all functions $g(\phi_{ABA'B'})$. Moreover, the family $\{\diff\nu_y(\phi_{ABA'B'}):y\in\rmD(\cH_{A})\}$ is $\diff\sigma_A(y)$-almost everywhere unique and each member $\diff\nu_y(\phi_{ABA'B'})$ is supported on $\tr_{BA'B'}^{-1}(y)$.

Without loss of generality, we only pay attention to full rank $y\in\rmD(\cH_A)$ because the set of rank-deficient density matrices $y$ has measure zero under $\diff\phi_{ABA'B'}$. Here, each fiber $\tr_{BA'B'}^{-1}(y)$ has been identified with the space $\sP\sC$. Under this identification, we will show that $\diff\nu_y(\phi_{ABA'B'})$ is almost everywhere equivalent to with the uniform measure on channel space $\diff \nu(\Psi)$. This follows from unitary invariance of $\diff\nu_y(\phi_{ABA'B'})$ and the uniqueness of the Haar measure $\diff U_{BA'B'}$. Specifically, since $\diff (U_{BA'B'}^\dagger\phi_{ABA'B'}U_{BA'B'}) = \diff \phi_{ABA'B'}$ for all $U_{BA'B'}$ we have by change of variables
\begin{align}
\int \diff \phi_{ABA'B'} g(\phi_{ABA'B'}) &= \int \diff (U_{BA'B'}^\dagger\phi_{ABA'B'}U_{BA'B'}) g(\phi_{ABA'B'}) \\
&= \int \diff \phi_{ABA'B'} g(U_{BA'B'}\phi_{ABA'B'}U_{BA'B'}^\dagger) \\
&= \int \diff\sigma_A(y)\int_{\tr_{BA'B'}^{-1}(y)}\diff\nu_y(\phi_{ABA'B'}) g(U_{BA'B'}\phi_{ABA'B'}U_{BA'B'}^\dagger) \\
&= \int \diff\sigma_A(y)\int_{\tr_{BA'B'}^{-1}(y)}\diff\nu_y(U_{BA'B'}^\dagger\phi_{ABA'B'}U_{BA'B'}) g(\phi_{ABA'B'}) \\
\end{align}
where the last equality follows from the fact that the fiber $\tr_{BA'B'}^{-1}(y)$ is invariant under all $U_{BA'B'}$. By uniqueness of the family, we have
\begin{align}
\diff\nu_y(U_{BA'B'}^\dagger\phi_{ABA'B'}U_{BA'B'}) = \diff\nu_y(\phi_{ABA'B'}) \textrm{ for all } U_{BA'B'}\,.
\end{align}
This says that each member $\diff\nu_y(\phi_{ABA'B'})$ of the disintegration family is unitary invariant. Due to the uniqueness of the normalized Haar measure we conclude $\diff\nu_y(\phi_{ABA'B'}) = \diff\nu(\Psi)$. In fact, we obtain correspondences between the objects
\begin{align}
\tr_{BA'B'}^{-1}(y) &\leftrightarrow \sP\sC \\
\diff\nu_y(\phi_{ABA'B'}) &\leftrightarrow \diff\nu(\Psi)
\end{align}
induced by $J_y^{-1}$.

Taking partial trace of system $A'B'$ yields the statement of the proposition and completes the proof.
\end{proof}

\bibliography{ChannelRegions,Tomography}

\end{document}